\documentclass[11pt]{article}
\usepackage[margin=1in]{geometry}
\usepackage{graphicx} 

\usepackage{caption} 
\usepackage{subfigure}
\DeclareCaptionStyle{ruled}{labelfont=normalfont,labelsep=colon,strut=off} 

\usepackage{algorithmic}
\usepackage{hyperref}

\usepackage[T1]{fontenc}
\usepackage{lmodern}\usepackage[T1]{fontenc}
\usepackage{lmodern}

\usepackage{url}            
\usepackage{booktabs}       
\usepackage{amsfonts}       
\usepackage{nicefrac}       
\usepackage{microtype}      
\usepackage{xcolor}         

\usepackage{amsmath,amssymb,amsthm}
\usepackage{thmtools}
\usepackage{thm-restate}

\newtheorem{theorem}{Theorem}

\newtheorem{lemma}[theorem]{Lemma}
\newtheorem{definition}[theorem]{Definition}
\usepackage{graphicx}
\usepackage{censor}

\usepackage{cuted}

\usepackage{comment}
\DeclareMathOperator*{\argmax}{arg\,max}

\title{Stochastic Multiplicative Weights Updates in Zero-Sum Games}  

\makeatletter
\renewcommand\@date{{%
  \vspace{-\baselineskip}%
  \large\centering
  \begin{tabular}{@{}c@{}}
    James P. Bailey \\
    \normalsize jamespbailey@tamu.edu \\
    \normalsize Texas A\&M University \\ \normalsize\phantom{hi}
  \end{tabular}%
  \quad  \quad
  \begin{tabular}{@{}c@{}}
    Sai Ganesh Nagarajan \\
    \normalsize sai.nagarajan@epfl.ch \\
    \normalsize EPFL \\ \normalsize\phantom{hi}
  \end{tabular}
  \quad   \quad
  \begin{tabular}{@{}c@{}}
    Georgios Piliouras \\
    \normalsize georgios@sutd.edu.sg \\
    \normalsize Singapore University of \\ \normalsize Technology and Design
  \end{tabular}

  \bigskip

}}
\makeatother

\begin{document}
\maketitle

\begin{abstract}
  We study agents competing against each other in a repeated network zero-sum game while applying the multiplicative weights update (MWU) algorithm with fixed learning rates. In our implementation, agents select their strategies probabilistically in each iteration and update their weights/strategies using the realized vector payoff of all strategies
  , i.e., stochastic MWU with full information.   We show that the system results in an irreducible Markov chain where agent strategies diverge from the set of Nash equilibria.   Further, we show that agents will play pure strategies with probability 1 in the limit.   
\end{abstract}

    	\def\size{.20}
    	\def\Space{-.1in}
    	\begin{figure}[!h]\centering
    		{		
    			\subfigure[10 Iterations]{\label{fig:10}%
    				\includegraphics[scale=\size]{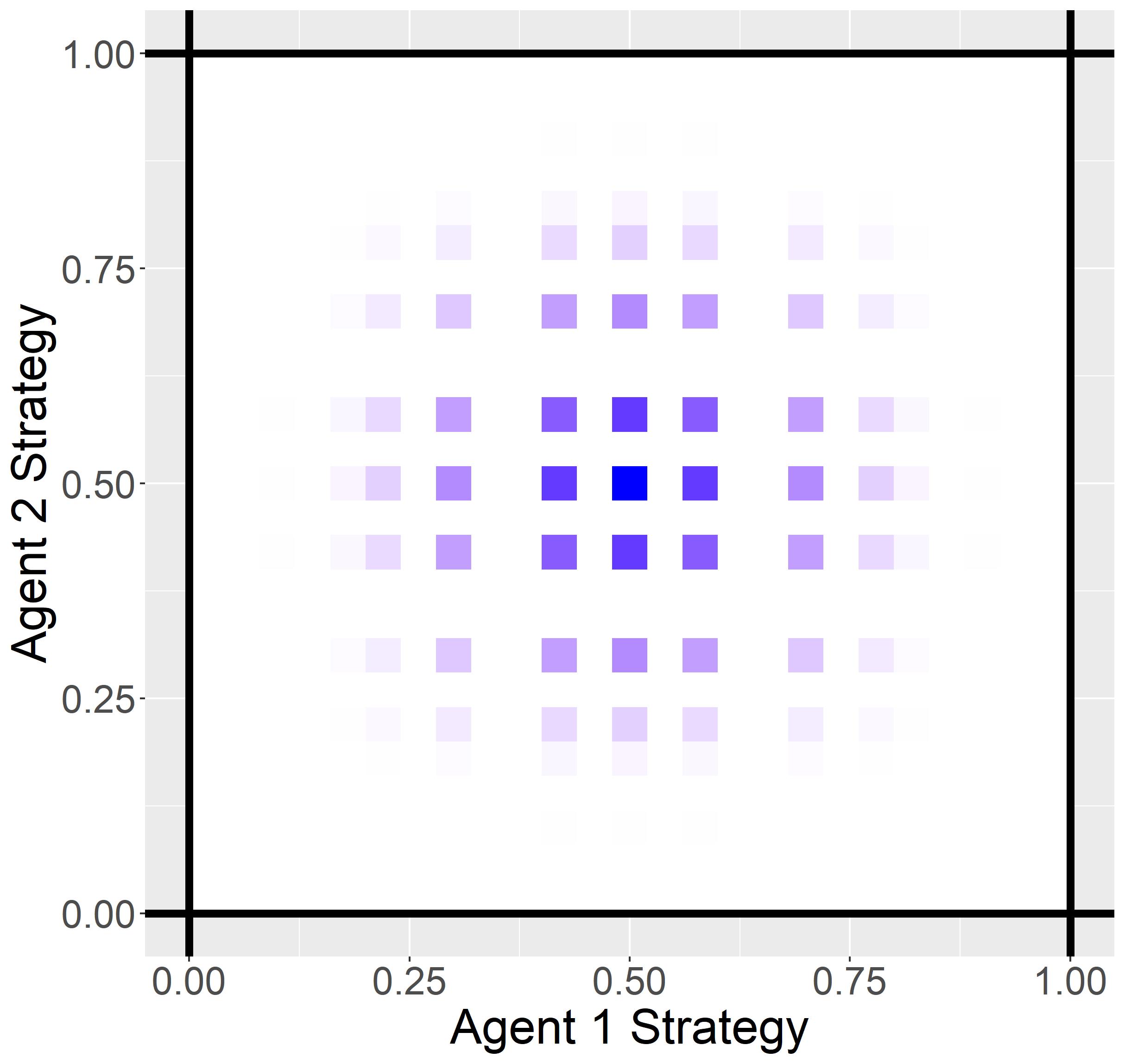}}
    			\hspace{\Space}
    			\subfigure[100 Iterations]{\label{fig:100}%
    				\includegraphics[scale=\size]{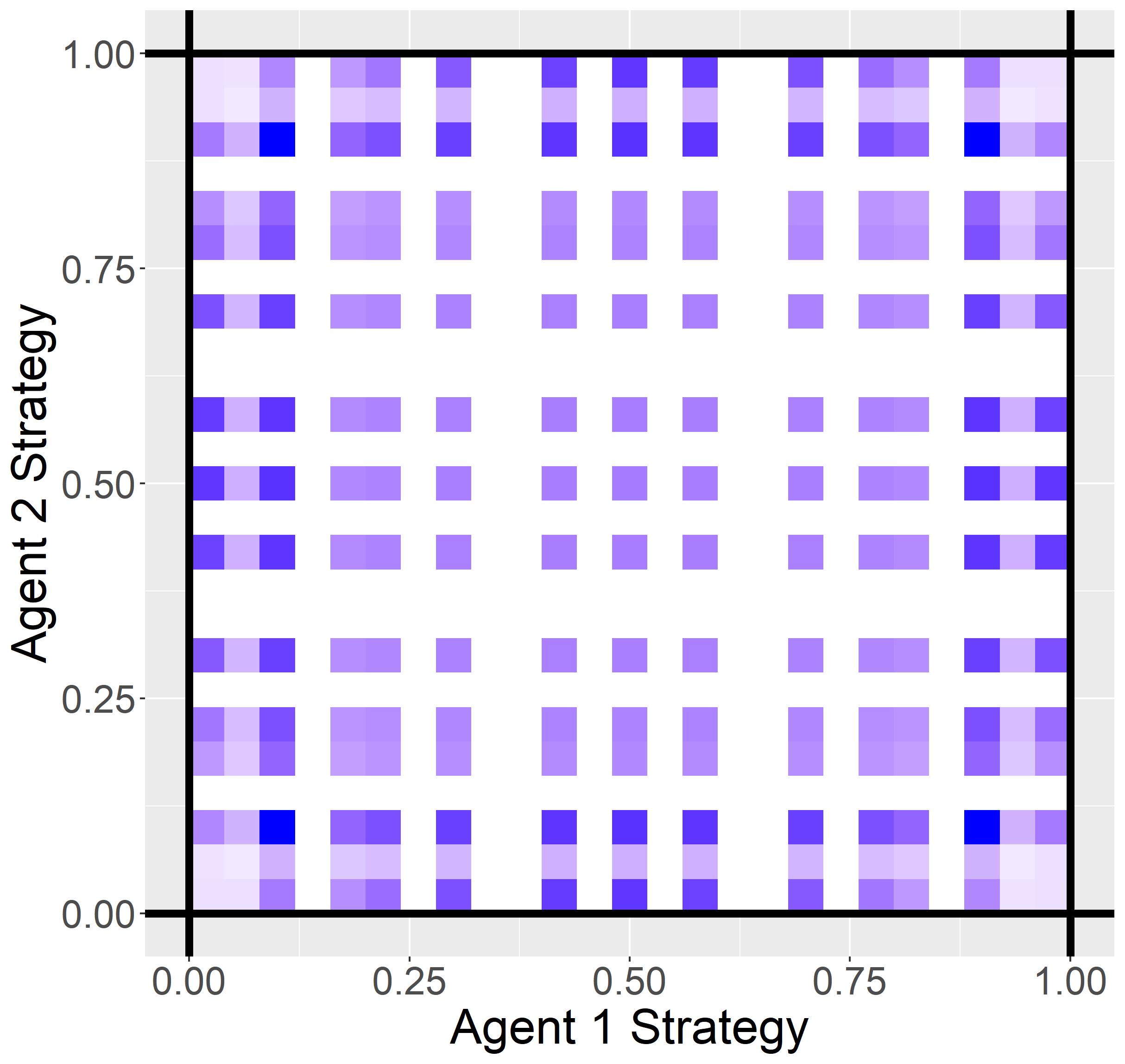}}
    			\hspace{\Space}
    			\subfigure[250 Iterations]{\label{fig:250}%
    				\includegraphics[scale=\size]{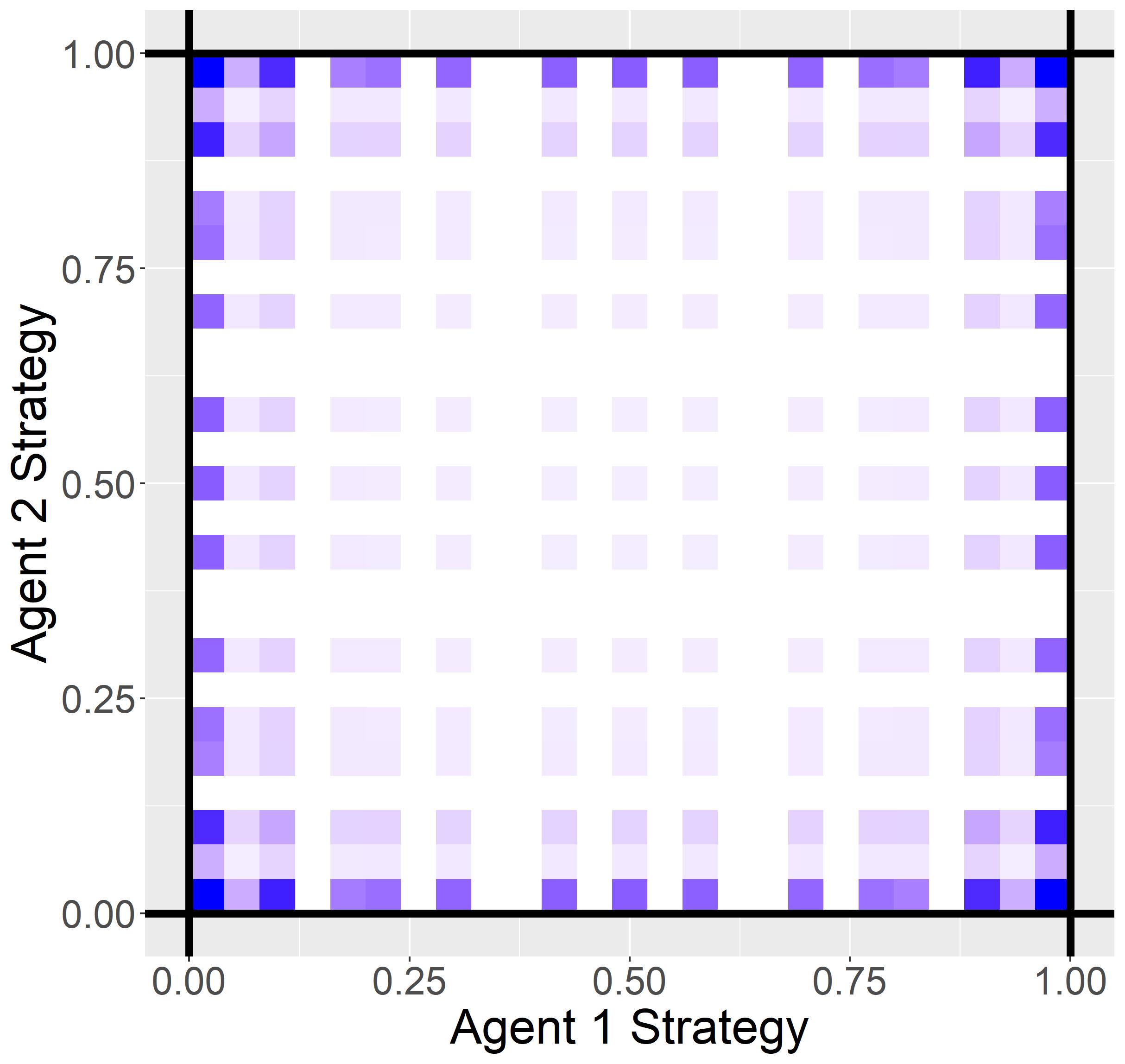}}
    			\hspace{\Space}
    			\subfigure[500 Iterations]{\label{fig:500}%
    				\includegraphics[scale=\size]{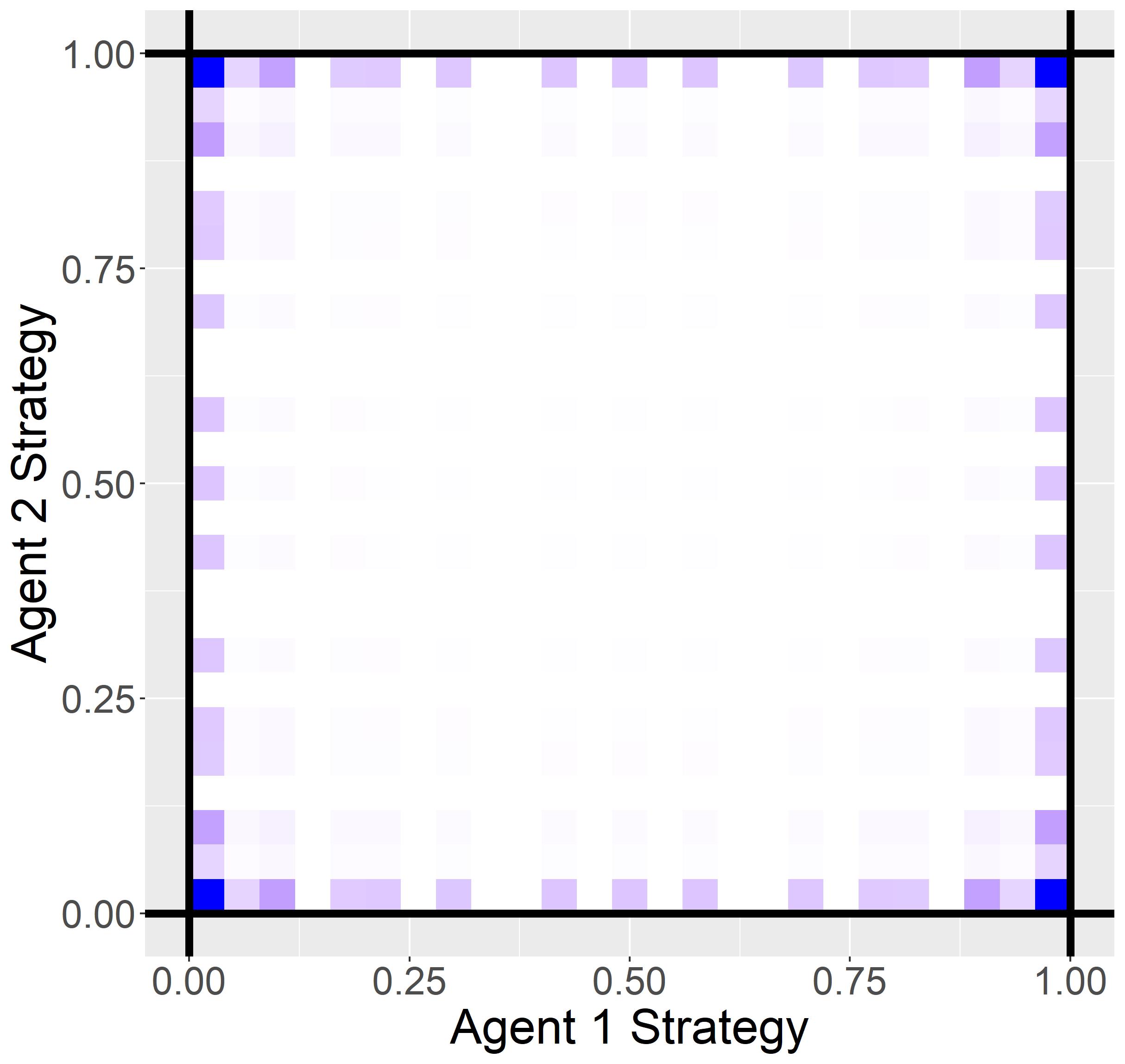}}
    		}  
    	{\caption{The Game Matching Pennies Updated with Stochastic MWU. The -axis and y-axis Show Agent 1's and Agent 2's Probabilities of Playing ``Heads'' Respectively.  The Opaqueness of Each Rectangle is Proportional to the Probability that the Agents' Strategies Appear in the Region. The Four Figures Demonstrate that Stochastic MWU Dissipates from the Equilibrium and to the Boundary -- Specifically to the Extreme Points (Pure Strategies) of the Region. In this Paper, we Prove These Strategies Converge to the Pure Strategies (Theorem \ref{thm:Convergence}). \label{fig:Opening}}}
    	\end{figure}

\section{Introduction}

Zero-sum games are arguably the most well studied class of social interactions within game theory.
At the same time one of the most well known results in online learning in games is that regret-minimizing algorithms such as Multiplicative Weights Update (MWU) and Follow-the-Regularized-Leader (FTRL) converge in a time-average sense to Nash equilibria \cite{Cesa06,freund1999adaptive}.

Recent results, however, have started to reveal a more intricate and detailed picture about the \textit{day-to-day behavior} of the dynamics by taking a dynamical systems approach, that is orthogonal to the typical regret approach, and by exploiting insights from  continuous-time dynamics and differential equations~\cite{piliouras2014optimization,GeorgiosSODA18}. \cite{BaileyEC18} showed that all FTRL dynamics, including MWU, GD \textit{diverge away} from the maxmin equilibrium in deterministic settings. In fact, they do so chaotically with small perturbations to the initial conditions leading to quickly diverging orbits~\cite{cheung19a,cheung2020chaos}.

The above chaotic instability results seem to paint a rather bleak picture when it comes to developing a common sense understanding of how these dynamics actually behave in practice. That is no long term predictions are meaningfully possible. Furthermore, all the work above focused on \textit{deterministic} dynamical systems, describing the expected behavior of learning dynamics. In reality, all regret-minimizing algorithms are randomized depending on the stochastic sampling of agents actions which only introduces a further source on uncertainty in an already dynamically complex system. This leads in to our central question.

\textit{ 
Is it possible to understand
the behavior
of stochastic MWU (and other FTRL dynamics) 
 in zero-sum games beyond merely stating negative, instability or unpredictability results? How do the dynamics actually behave?}



At a first glance it may seem rather surprising that despite the classic nature of MWU~\cite{AHK2012}, its day-to-day behavior in most standard of game theoretic settings, zero-sum games, has not be analyzed before. Indeed, MWU
has been rediscovered many times either in its exact form on in numerous closely connected variants, found in \cite{Cesa06, fudenberg1998theory, littlestone1994weighted, vovk1998game, Wei95, freund1997decision}. The same applies of course for general FTRL dynamics, arguably the most well known class of regret-minimizing dynamics and a staple of online optimization theory~\cite{hazan2016introduction}.
Nevertheless, so far the focus on the analysis of such algorithms was in understanding their regret properties, whose convergence to zero in competitive games immediately implies time-average convergence to Nash. Understanding their day-to-day stochastic behavior, as we show, requires a combination of non-trivial tools and techniques spanning convex optimization (Bregman divergence, Fenchel couplings), Markov chains in countable state spaces and Feller chains in uncountable state spaces, dynamical systems (Lyapunov theory) and game theory.

{\bf Our results and techniques.} 
We establish that stochastic variants of Follow-the-Regularized-Leader (FTRL) with fixed learning rates result in agents that almost always play strategies close to the boundary in the setting of network zero-sum games (Theorem \ref{thm:boundary}).
We accomplish this by formally showing that Stochastic FTRL induces an irreducible Markov chain where each iteration of FTRL causes agents to move away from the set of Nash equilibrium in expectation. 
The evolution of this Markov chain is depicted in Figure \ref{fig:Opening} with the blue regions (strategies) diverging from Nash and to the extreme points. 

In our key technical result, we show that in the setting of 2-agent zero-sum games where agents use MWU, every convergent subsequence of agent strategies must converge to a mixture of pure strategies, i.e., {\bf agents spend almost all of their time playing effectively pure strategies} (Theorem \ref{thm:Convergence}). 
We remark that this result is substantially stronger than the divergence result known for deterministic MWU \cite{BaileyEC18}, which only shows strategies converge to the boundary of the strategy space. 
 Given that randomized strategies are the normative solution concepts for zero-sum games (e.g. Matching Pennies, Rock-Paper-Scissors), we showcase a maximal disagreement between the predictions of Nash equilibrium (``defensive" maxmin play;  trying to minimize potential lossless by being unpredictable to the opponent) and the actual behavior of learning dynamics in practice (``strong-headed" behavior; playing with full confidence strategies than can be exploited by the opponent which happen to currently have good historical returns).
We establish this result by by constructing a Feller chain and proving that the only stationary points of the dynamics are pure strategies (Theorems \ref{thm:Feller} and \ref{thm:Stationary}). 
Theorem \ref{thm:Convergence} then follows from known results in Markov theory.
Our results showcase the value of introducing techniques related to ergodic theory~\cite{cornfeld2012ergodic}, where the object of study are the statistical properties of system trajectories in the understanding of online learning, optimization and game theory. 


\section{Preliminaries \& Model} \label{sec:notation}

\subsection{Normal Form Games}
A finite normal-form game $\Gamma  \equiv \Gamma({\cal N}, {\cal S}, A)$ consists of a set of agents ${\cal N}=\{1,...,N\}$ where agent $i$ may select from a finite set of actions or pure strategies ${\cal S}_i=\{1,...,{S}_i\}$.
Given the set of actions $s\in {\cal S}:=\times_{i\in {\cal N}} {\cal S}_i$, agent $i$ receives the payout $\sum_{j \neq i} A^{(ij)}_{s_i,s_j}$ where $A^{(ij)}$ is the \emph{payoff matrix} between agents $i$ and $j$. 
To simplify notation, we let $e_{s_i}$ be the standard basis vector where the $s_i$th coordinate of $e_{s_i}$ is 1 and all other coordinates are 0. 
With this notation, we denote $i$'s payout as $\langle e_{s_i},\sum_{j \neq i} A^{(ij)}e_{s_j}\rangle$.

In this paper, we study only \emph{non-trivial games} -- games where there is at least one agent $\{i,j\}$ and a pair of strategies $s, \bar{s}\in {\cal S}$ such that $\sum_{j\neq i} A^{(ij)}e_{s_j}\neq \sum_{j\neq i} A^{(ij)}e_{\bar{s}_j}$.
In a trivial game, agents' payouts are independent of the actions of other agents, every mixed strategy is a Nash equilibrium,  and thus the ``dynamics'' of the game are irrelevant.

Agents are also allowed to use mixed strategies $x_i= (x_{is_i})_{s_i\in {\cal S}_i}\in {\cal X}_i= \{ x_i\in \mathbb{R}_{\geq 0}^{S_i}:\sum_{s_i\in {\cal S}_i} x_{is_i}=1\}$, i.e., $x_i$ is a probability vector over the set of pure strategies.
A strategy is fully mixed if $x_{is_i}>0$ for all $s_i\in {\cal S}_i$ and $i\in {\cal N}$ -- equivalently, $x_i\in rel.int({\cal X}_i)$. 
In practice, a mixed strategy describes a probability distribution over the set of pure strategies. Given a probability distribution $x_i$, agent $i$ selects strategy
$s_i \text{ with probability } x_{is_i}$.
Thus, given a set of mixed strategies $x=\times_{i\in {\cal N}}x_i$, agent $i$'s expected payout is $\langle x_i, \sum_{j\neq i} A^{(ij)} x_j\rangle$.

The most commonly used solution concept for games is the \emph{Nash equilibrium}. 
A Nash equilibrium (NE) is  a strategy $x^*\in {\cal X}$ where no agent can do better by deviating from $x_i^*$. 
Formally, 
\begin{align}\label{eqn:Nash} \langle x_i^*, \sum_{j\neq i} A^{(ij)} x_j^*\rangle \geq \langle x_i, \sum_{j\neq i} A^{(ij)} x_j^*\rangle \forall x_i\in {\cal X}_i,  i\in {\cal N} \tag{Nash Equilibrium}\end{align}

\subsubsection{Zero-Sum Games}
We specifically study network zero-sum games -- games where $A^{(ij)}=-[A^{(ji)}]^\intercal$ implying agent $i$ and agent $j$'s total utility for their interaction, i.e., $\langle x_i, A^{(ij)} x_j\rangle + \langle x_j, A^{(ji)} x_i\rangle$, is zero. 
This implies that the total utility gained from all agents is also zero, i.e., $\sum_{i\in {\cal N}}\langle x_i, \sum_{j\neq i} A^{(ij)} x_j\rangle =0$.

Within this paper, we assume that there exists a Nash equilibrium $x^*$ such that  $\langle x_i^*, \sum_{j\neq i} A^{(ij)} x_j^*\rangle=0$ for all $i\in {\cal N}$.
Notice that shifting $A^{(ij)}$ by a constant does not change the set of Nash equilibria in the game. 
Iteratively for $i=1,...,N-1$, $\{A^{(ij)}\}_{j=i+1}^N$ can be increased by a constant to ensure   $\langle x_i^*, \sum_{j\neq i} A^{(ij)} x_j^*\rangle=0$ without altering previous agents' payouts. 
Since the game is zero-sum, this immediately implies $\langle x_N^*, \sum_{j\neq N} A^{(Nj)} x_j^*\rangle=0$.

\subsection{Online Learning in the Deterministic Implementation of Mixed Strategies}

Rarely in the study of games do agents know the set of Nash equilibria, or even the utility function, prior to selecting their strategies. 
Rather, agents iteratively update their mixed strategies overtime based on the performance of pure strategies in prior iterations via an online learning algorithm.   
The most classical set of online learning algorithms are the Follow-the-Regularized-Leader (FTRL) algorithms, e.g., Gradient Descent, and Multiplicative Weights Update (MWU). 
Given a strictly convex regularizer $h_i: {\cal X}_i \to \mathbb{R}$, an agent updates their strategies via
\begin{equation}\label{eqn:D.FTRL}\tag{Deterministic FTRL}
\begin{aligned}
{y}_{i}^t &= {y}_{i}^{t-1}+\sum_{j\neq i}  A^{(ij)}x_j^{t-1};\\
{x}_i^t &= \argmax_{x_i \in {\cal X}_i} \left\{\left<y_i^t,x_i\right>-\frac{h_i(x_i)}{\eta_i}\right\}.
\end{aligned}
\end{equation}

The payoff vector $y_i^t$ represents the cumulative payout for any pure strategy since the beginning of the time.  
Formally, $y_{is_i}^t-y_{is_i}^0$ denotes the cumulative payout agent would have received has she played pure strategy $s_i$ from iteration $0$ to iteration $t-1$. 
Thus, agent $i$ selects the strategy $x_i^t$ that maximizes the difference between her cumulative payout since time 0 and a strictly convex regularization term. 

The learning rate is specified by $\eta_i$.
Since most of our results hold for general regularizers, we will often embed the learning rate into the regularizer and assume $\eta_i=1$ for all agents.

The two most well known variants of FTRL are Gradient Descent and MWU algorithms obtained via the regularizers $h_i(x_i)=||x_i||_2^2/2$ and $h_i(x_i)=\sum_{s_i\in {\cal S}_i} x_{is_i}\ln x_{is_i}$ respectively. 
By iteratively solving \ref{eqn:D.FTRL}, (MWU) can be written as 
\begin{equation}\label{eqn:D.MWU}\tag{Deterministic MWU}
\begin{aligned}
x_{is_i}^t&=\frac{x_{is_i}^{t-1}\exp{(\eta_i\cdot  \sum_{j\neq i} e_{s_i}A^{(ij)}x_j^{t-1})}}{\sum_{\bar{s}_i\in {\cal S}_i}x_{i\bar{s}_i}^{t-1}\exp{(\eta_i\cdot  \sum_{j\neq i} e_{\bar{s}_i}A^{(ij)}x_j^{t-1})}}.
\end{aligned}
\end{equation}

Another important function we use in our analysis of (\ref{eqn:D.FTRL}) is the convex conjugate  $h^*_i:\mathbb{R}^{|{\cal S}_i|} \mapsto R$ given by 
$h^*_i(y_i)=\sup_{x_i \in \mathcal{X}_i} \left\{\left<y_i,x_i\right>-h_i(x_i)\right\}$.
Specifically, $h^*_i$ establishes a duality between the mixed strategy $x_i^t$ and the payoff vector $y_i^t$. 
The \emph{maximizing argument} \cite{HKSS12}, a well known property of FTRL, establishes the connection $x_i^t=\nabla h^*_i(y_i^t)$ when $\eta_i=1$.

\subsection{Stochastic FTRL}\label{sec:Stochastic}
In the definition of (\ref{eqn:D.FTRL}) and (\ref{eqn:D.MWU}) we assume that the mixed strategies are implemented deterministically and therefore the update rules are deterministic. 
In practice however, the mixed strategies denote a probability distribution over the set of pure strategies and the realized strategies are determined randomly. 
In this section, we extend the definitions of FTRL and MWU to stochastic implementations of the set of mixed strategies.

In this setting, both the cumulative payoff vectors $y^t$ and mixed strategies $x^t$ are random variables. 
We use standard notation from probability where the uppercase $Y^t$ denotes the probability mass function for agents' cumulative payoff vectors and $X^t$ to denote the probability mass function used to select agents' pure strategies in iteration $t$.
\begin{equation}\label{eqn:S.FTRL}\tag{Stochastic FTRL}
\begin{aligned}
{Y}_{i}^t &={Y}_{i}^{t-1} +\sum_{j\neq i} A^{(ij)}s_j \ \text{with  probability} \ \prod_{j\neq i} X_{js_j}^{t-1}\\
{X}_i^{t} &= \argmax_{x_i \in \mathcal{X}} \left\{\left<Y_i^{t},x_i\right>-\frac{h_i(x_i)}{\eta_i}\right\}\\
\end{aligned}
\end{equation}  

Formally, we actually work with the dynamic ${Y'}_{i}^t={Y}_{i}^t -\mathbf{1}\cdot Y_{i1}^t$ where $\mathbf{1}$ is a vector of 1s so that the first component of $Y_i'$ is always 0. In the definition of (\ref{eqn:S.FTRL}), subtracting a constraint from $Y_i^t$ since $x_i$ is a probability vector and shifting $Y_i^t$ by a constant just shifts $\langle Y_i^t,x_i\rangle$ by a constant. 
This distinction allows us to establish a bijection between $X$ and $Y$. 

The standard regret proof for FTRL holds when your opponent updates in any fashion, including randomly, and therefore extends to stochastic implementations of FTRL.  
However, we know of no results that examine the actual dynamics of stochastic implementations of FTRL.
Typically, we expect a random variable not to deviate too much from its expectation.
Indeed, by linearity of expectation, 
\begin{align*}
E[Y^t_{i}|Y^{t-1}=y^{t-1}]
=y^{t-1}_{i}+\eta_i \cdot \sum_{j\neq i} A^{(ij)}x^{t-1}_{i}
=y^t_{i}
\end{align*}
where $x_i^{t-1}=\nabla h^*_i(y_i^{t-1})$ and $y^t$ are found the (\ref{eqn:D.FTRL}). 
However, the dynamics of (\ref{eqn:D.FTRL}) can be drastically different from (\ref{eqn:S.FTRL}). 
For instance, (\ref{eqn:D.FTRL}) results in a stationary strategy when $x^0$ is a Nash equilibrium while we show later that (\ref{eqn:S.FTRL}) tends to the boundary of the strategy space regardless of the initial strategy.

In particular, we are interested in the stochastic version of (\ref{eqn:D.MWU}) given by: 
\begin{equation}\label{eqn:S.MWU}\tag{Stochastic MWU}
\begin{aligned}
X_{is_i}^t&=\frac{X_{is_i}^{t-1}\exp{(\eta_i\cdot  \sum_{j\neq i} e_{s_i}A^{(ij)}X_j^{t-1})}}{\sum_{\bar{s}_i\in {\cal S}_i}X_{i\bar{s}_i}^{t-1}\exp{(\eta_i\cdot  \sum_{j\neq i} e_{\bar{s}_i}A^{(ij)}X_j^{t-1})}}.
\end{aligned}
\end{equation}
which is obtained from (\ref{eqn:S.FTRL}) using the regularizer $h_i(x)=\sum_{s_i\in {\cal S}_i} x_{is_i}\ln x_{is_i}$. 
We formally show this equivalence in Appendix \ref{app:Noninterior}.

\subsection{Bregman Divergence from a Nash Equilibrium} 

To establish that the mixed strategies tend to the boundary, we work with a ``distance'' between the agents' strategies and an interior Nash equilibrium.
The general idea is that if the distance to an interior Nash equilibrium is large enough, then the agents' strategies must be close to the boundary. 
The standard notion of distance used when studying online learning algorithms is the Bregman divergence.


When analyzing (\ref{eqn:S.FTRL}) with regularizer $h_i$, we study the Bregman divergence with regularizer $h_i$. 
\begin{align}D_{h}(x^*||x)
=\sum_{i\in{\cal N}}\left( h_i(x_i^*)-h_i(x_i)- \langle \nabla h_i(x_i), x_i^*-x_i\rangle \right)\tag{Bregman Divergence}
\end{align}
For (\ref{eqn:D.MWU}) and (\ref{eqn:S.MWU}), the Bregman divergence is referred to as the Kullback-Leibler (K-L) divergence and is given by
\begin{align}D_{KL}(x^*||x)
=\sum_{i\in{\cal N}}\sum_{s_i\in {\cal S}_i}x^*_{is_i}\left(\ln{x^*_{is_i}}-\ln{x_{is_i}}\right).\label{eqn:DefKL}\tag{K-L Divergence}
\end{align}

Another notion of distance we use to understand the dynamics is the Fenchel-coupling that measures the distance from agents' strategies to the Nash equilibrium in the space of payoff vectors. 
The Fenchel-coupling is given by
\begin{align}\label{eqn:Fenchel}
F_h(x^*||y)=\sum_{i \in {\cal N}}\left(h_i(x^*_i)+h^*_i(y_i)-\langle y_i,x^*_i \rangle \right).\tag{Fenchel-coupling}
\end{align}

The Fenchel-coupling and Bregman divergence are closely related.  
Formally, $F_h(x^*||y^t)\geq D_h(x^*||x^t)$ where equality holds whenever $x^t$ is fully mixed (see \cite{GeorgiosSODA18, Mer16} and Lemma \ref{lem:FBequiv}). 
In (\ref{eqn:S.FTRL}), we study the random variables $D_h(x^*||X^t)$ and $F_h(x^*||Y^t)$ to establish that that agents tend to select mixed strategies close to the boundary. 

We also remark that shifting $y$ by a constant does not change the Fenchel-coupling for our particular dynamics. 
Since $x$ in the definition of $h^*_i$ is a probability vector, increasing $y$ by $c\cdot \mathbf{1}$ causes $h^*_i$ to increase by $c$ while $-\langle y_i,x_i^*\rangle$ decreases by $c$ and as a result, the Fenchel-coupling does not change with constant shifts to $y_i$.

\subsection{Markov Chain Basics}

In (\ref{eqn:S.FTRL}), the random variable $Y^t$ depends only on the value of the same random variable in the previous iteration, i.e., ${Y}_{i}^t ={Y}_{i}^{t-1} +\sum_{j\neq i} A^{(ij)}s_j \ \text{with  probability} \ \prod_{j\neq i} X_{js_j}^{t-1}=\prod_{j\neq i} \nabla h^*_{js_j}({Y}_{j}^{t-1})$.
These types of memory-less properties are frequently modeled as Markov chains. In this section, we introduce the notation necessary to understand (\ref{eqn:S.FTRL}) as a Markov chain.  
We introduce the notation with respect to the random variables and strategy spaces introduced in the previous sections. 
We take ${\cal B}({\cal X})$ to be the Borel $\sigma$-algebra on ${\cal X}$. 

\begin{definition}[Transition Probability Kernel]
A deterministic function $P:\mathcal{X} \times \mathcal{B}(\mathcal{X}) \mapsto \mathbb{R}^{+}$ is a Transition Probability Kernel if
\begin{itemize}
\item {for each} $A \in \mathcal{B}(\mathcal{X}), P( \cdot ,A)$ {is a non-negative measurable function on} $\mathcal{X}$
\item {for each} $x \in \mathcal{X}, P(x, \cdot )$ is a probability measure on $\mathcal{B}(\mathcal{X})$.
\end{itemize}
\end{definition}

\begin{definition}
	The probability measure $\pi$ is invariant (or stationary) with respect to $P$ if $\pi P = \pi$, i.e.,
	\begin{align*}
		\int_{\cal X} \pi(dx) P(x,A) = \pi (A) \ \text{ for all } A\in {\cal B}({\cal X})\tag{Stationarity Condition}
	\end{align*} 
\end{definition}

\begin{definition}\label{def:Feller}
	A Markov Chain defined on a metric space ${\cal X}$ is said to be a Feller chain if $P(x^{(n)}, \cdot) \Rightarrow P(x,\cdot)$ as $x^{(n)} \to x$, i.e., $P(x^{(n)},\cdot)$ converges weakly (or in distribution) to $P(x,\cdot)$ as $x^{(n)} \to x$. 
\end{definition}

\section{Convergence to the Boundary}\label{sec:converge}

We begin by showing that for any set in the interior of ${\cal X}$, that once the strategies leave this set then we expect an infinite number of iterations to pass before returning whenever there is a fully-mixed Nash equilibrium.
This implies that agents almost always play strategies arbitrarily close to boundary and FTRL results in extremal strategies.

\begin{restatable}[]{theorem}{Return}\label{thm:returntime}
	Let $\Gamma({\cal N}, {\cal S}, A)$ be any non-trivial network zero-sum game with a rational fully-mixed Nash equilibrium $x^*$ and let $B$ be any compact set in the interior of ${\cal X}$.  Suppose $\{X^t\}_{t=0}^\infty$ is updated according to (\ref{eqn:S.FTRL}) where $\nabla h^*_i(y_i)>{\bf 0}$ for all $y_i$. The expected time to return to the set $B$ after leaving the set $B$ is infinity. 
\end{restatable}

\begin{restatable}[]{theorem}{Boundary}\label{thm:boundary}
	Let $\Gamma({\cal N}, {\cal S}, A)$ be any non-trivial network zero-sum game with a rational fully-mixed Nash equilibrium $x^*$ and let $B$ be any compact set in the interior of ${\cal X}$.  Suppose $\{X^t\}_{t=0}^\infty$ is updated according to (\ref{eqn:S.FTRL}) where $\nabla h^*_i(y_i)>{\bf 0}$ for all $y_i$. The proportion of strategies $\{X^t\}_{t=0}^T$ where $X^t\in B$ goes to 0 as $T\to \infty$ almost surely. 
\end{restatable}

To establish these results, we first show that the strategies are expected to move away from the set of Nash equilibria (Theorem \ref{thm:Drift} in Appendix \ref{app:BoundaryProofs}).  
Specifically, we show that, in expectation, that the Fenchel-coupling in the dual-space of payoff vectors is increasing. 
Next, in Section \ref{sec:dual}, we introduce a Markov chain to describe the behavior of (\ref{eqn:S.FTRL}) and show that it is irreducible (Theorem \ref{thm:Irreducible}).
Along with Theorem \ref{thm:Drift}, Theorems \ref{thm:returntime} and \ref{thm:boundary} then follow readily from well known results in Markov theory. 
The full details can be found in Appendix \ref{app:BoundaryProofs}. 

Theorem \ref{thm:boundary} implies that agents converge to strategies on the boundary of ${\cal X}$, i.e., there is almost always an agent $i$ playing a strategy $s_i \in {\cal S}_i$ with a probability close to 0. 
Thus, despite agents strategies having a time-average convergence to the set of approximate Nash equilibria, the strategies are actually repelled from the set of Nash equilibria and agents select extreme strategies.

We also show that strategies converge to the boundary when there is not a fully-mixed Nash equilibrium in 2-agent variants of (\ref{eqn:S.MWU}). 
Like Theorem \ref{thm:boundary}, Theorem \ref{thm:noninterior} shows that agents will rarely play fully-mixed strategies.

\begin{restatable}[]{theorem}{Noninterior}\label{thm:noninterior}
	For almost every 2-agent zero-sum game with a unique Nash equilibrium on the boundary, there exists an $\eta_0$ such that for all $\eta < \eta_0$ agent strategies will converge to the boundary with probability 1 when agents use (\ref{eqn:S.MWU}).
\end{restatable}

The proof of Theorem \ref{thm:noninterior} follows similarly to the case of (\ref{eqn:D.MWU}) from \cite{BaileyEC18} and is deferred to Appendix \ref{app:Noninterior}. 
We remark that this result likely extends to all of (\ref{eqn:S.FTRL}), with arbitrary learning rates, and with multiple agents. 
However, several new techniques need to be developed in order to extend much of the analysis involving non-interior Nash to more general settings. 

Together, Theorems \ref{thm:boundary} and \ref{thm:noninterior} imply that (\ref{eqn:S.MWU}) converges to the boundary in every non-trivial 2-agent zero-sum game.

\subsection{Constructing a Markov Chain in the Dual-space of Payoff Vectors}\label{sec:dual}

We begin by constructing the state space for the underlying Markov chain in the dual-space of the agent payoff vectors used in (\ref{eqn:S.FTRL}). 
\begin{equation}\tag{States Reachable After $t$ Iterations}
\begin{aligned}
&{\cal Y}^0 = y_i^0 \\
&{\cal Y}^t= \bigcup_{y\in {\cal Y}^{t-1}} \bigcup_{{s\in {\cal S}}} (y_1 + \sum_{j\neq 1}A^{(1j)}e_{s_j}, \cdots, y_{N} + \sum_{j\neq {N}}A^{({N}j)}e_{s_j} )
\end{aligned}
\end{equation}
In this definition, ${\cal Y}^t$ denotes the possible payoff vectors for all agents after $t$ iterations of (\ref{eqn:S.FTRL}). 
In particular, if agent $i$ has the payoff vector $y_i^{t-1}$ in iteration $t-1$, and if the agents randomly select the pure strategies $s$ in iteration $t-1$, then agent $i$'s payoff vector in iteration $t$ will be $y_i^t=y_i^{t-1} + \sum_{j\neq i}A^{(ij)}e_{s_j}$.
This yields the following transition probabilities:
\begin{align*}
\bar{P}(y^{t-1},Y^t=y^t)&= \sum_{{s\in S:} \atop {y_i^t ={y}_{i}^{t-1} +\sum_{j\neq i} A^{(ij)}s_j} \ \forall_{i=1}^N} \prod_{j=1}^N x_{js_j}^{t-1}\tag{Transition Kernel in Dual-space}
\end{align*}
where ${x}_i^{t-1} = \argmax_{x_i \in \mathcal{X}} \left\{\left<y_i^{t-1},x_i\right>-\frac{h_i(x_i)}{\eta_i}\right\}$ is the realization of $X_i^{t-1}$ as given in the definition of (\ref{eqn:S.FTRL}).
With this definition, the probability of going from state $y_i^{t-1}$ to state $y_i^t$ is 0 for most states.  The probability is positive only if there is a set of realizable strategies $s\in S$ such that $y_i^t= y_i^{t-1}+\sum_{i\neq j} A^{(ij)}s_j$ for all $i=1,...,N$.
We also remark that we can normalize each $Y$ so that $Y_{i1}=0$ for each agent without changing the proof of irreduciblity. 
Moreover, this normalization will be useful later for establishing a bijection between the primal and dual-spaces. 

\begin{restatable}[]{theorem}{Irreducible}
	\label{thm:Irreducible}
	Let $\Gamma( \cal{N}, S, A)$ be any network zero-sum game with a rational fully-mixed Nash equilibrium $x^*$, such that $\sum_{j \neq i} A^{(ij)}x^*_j= {\bf 0}$ for all agents $i$ and $j$.  
	If the regularizer used in (\ref{eqn:S.FTRL}) satisfies $\nabla h_i^*(y_i)>{\bf 0}$ for all $y\in {\cal Y}$ and $i\in {\cal N}$, then the Markov Chain with state space ${\cal Y}$ and transition probabilities $\bar{P}$ is irreducible. 
\end{restatable}

The condition $0=\sum_{i\neq j} A^{(ij)}x^*_j$ can be assumed without loss of generality by adding constants to each agents' payoff matrices.  
Since strategies are probability vectors, this will shift agent payouts by a constant and make no difference to the strategies realized by (\ref{eqn:S.FTRL}). 

Irreducibilty requires that each state can be reached from any other state after some number of steps. 
In our proof, we show that every state can be reached from $y^0$ and that every state can return to $y^0$.  
The first part follows by construction of the state space and since $X_i^t=\nabla h_i^*(Y_i^t)>{\bf 0}$ implies any pure strategy can be played in any iteration. 
Denote ${\cal Y}^t$ as the set of states that can be reached after $t$ iterations. 
The second part relies on the rationality of $x_{is_i}^*:=c_{is_i}/b$  where $c_{is_i},b\in \mathbb{Z}_{>0}$, and $0=\sum_{i\neq j} A^{(ij)}x^*_j=\sum_{i\neq j}  A^{(ij)}c_{i}/b$ to construct a sequence of $b-1$ realized strategies to move from a state in ${\cal Y}^t$ to a state in ${\cal Y}^{t-1}$. 
Inductively, this implies every state can eventually reach $y^0$ and therefore the Markov chain is irreducible.  The proof is
in Appendix \ref{app:Irreducible}.

The irreducibilty of this Markov chain allows us to use several theorems from Markov theory to assist in proving Theorems \ref{thm:returntime} and \ref{thm:boundary}.

\section{Pure Strategies Almost Always}\label{sec:corners}

The results of the previous section show that the behavior of (\ref{eqn:S.FTRL}) is similar to the behavior of (\ref{eqn:D.FTRL}) shown in \cite{BaileyEC18}; agent strategies converge to the boundary of ${\cal X}$ in both cases. 
However, in this section, we show a significantly stronger result; 
agent strategies gravitate toward the extreme points (pure strategies) of ${\cal X}$.

\begin{restatable}[]{theorem}{Convergence}\label{thm:Convergence}
	Let $\Gamma({\cal N}, {\cal S}, A)$ be any 2-agent zero-sum game where every element of $A=A^{(12)}$ is unique and let
	$\bar{X}^T=\sum_{t=1}^T X^t/T$ where $X^t$ is generated according to  (\ref{eqn:S.MWU}). 
	There exists a $\eta_0$ such that for all $\eta<\eta_0$,
	\begin{enumerate}
		\item the sequence $\{\bar{X}^t\}_{t=1}^\infty$ has a convergent subsequence.
		\item  Every convergent subsequence of $\{\bar{X}^t\}_{t=1}^\infty$ converges to a mixture of pure strategies. 
	\end{enumerate}
\end{restatable}

Theorem \ref{thm:Convergence} relies on Theorem \ref{thm:noninterior} in that it requires agent strategies to converge to the boundary even when there is not an interior Nash equilibrium. As mentioned in Section \ref{sec:converge}, Theorem \ref{thm:noninterior} likely extends for arbitrary learning rates and for network zero-sum games.  
Once this generalization is made, Theorem \ref{thm:Convergence} also immediately extends for arbitrary learning rates and for network zero-sum games.  

Further, almost every matrix satisfies the requirement that the elements of $A$ are unique.  
The requirement that the elements of $A$ are unique ensure that the game induced on every face of ${\cal X}$ will be non-trivial,
i.e., given $\bar{\cal S}_i\subseteq {\cal S}_i$, the game played on the face $\bar{\cal X}=\{x\in {\cal X}: x_{is_i}=0 \ \forall s_i\in \bar{\cal S}_i \ \forall i=1,...N\}$ will be non-trivial. 
Without this requirement, it is possible that a convergent subsequence of $\{\bar{X}^t\}_{t=1}^\infty$ converges to the interior of a face that induces a trivial game. 
Finally, we remark that requirement that $A$ has unique elements can be weakened to ``every row and every column of $A$ have distinct elements'' thereby allowing payoff matrices such as $A=\left[\begin{array}{r r} 1 & -1 \\ -1 & 1 \end{array}\right]$. 

The proof of Theorem \ref{thm:Convergence} follows by first constructing a new Markov chain in the primal space and showing that this new Markov chain forms a Feller chain. 
The construction of this Feller chain is given in Section \ref{sec:primal}. 
Feller chains have interesting properties in that they always have a convergent subsequence as in the statement of Theorem \ref{thm:Convergence}. 
Moreover, every convergent subsequence of a Feller chain must converge to an stationary distribution (see e.g., Theorem 12.3.2 in \cite{douc2018markov}).
Using Theorems \ref{thm:boundary} and \ref{thm:noninterior}, strategies in the interior of ${\cal X}$, and in the interior of any $d$-dimensional face ${\cal X}$ where $d\geq 2$, will converge to their respective boundaries implying the only stationary distribution will be a mixture of pure strategies. 
The proof of Theorem \ref{thm:Convergence} is given in full detail in Appendix \ref{app:convergence}.

\subsection{Constructing the Markov Chain in the Primal-space}\label{sec:primal}

The primal space Markov chain is mostly built from the dual-space Markov chain via $\nabla h^*: {\cal Y} \to {\cal X}$. 
In particular, for (\ref{eqn:S.MWU}), $\nabla h^*$ forms a bijection between ${\cal Y}$ and the relative interior of ${\cal X}$. This creates a natural Markov chain in the relative interior of the primal space. 
\begin{align*}
\hat{P}_{\cal X}(x^{t-1},X^t=x^t)=\bar{P}(y^{t-1},Y^t=y^t) \
\text{ where } \nabla h^*(y^{t-1})=x^{t-1} \text{ and } \nabla h^*(y^{t})=x^{t}
\end{align*}	
However, to show strategies converge to pure strategies, it will be useful to extend this Markov chain to all of ${\cal X}$, including the boundary. 
Along the boundary, there is not necessarily a unique mapping between ${\cal X}$ and ${\cal Y}$, e.g., gradient descent (\ref{eqn:S.FTRL}) with $h_i(x_i)=||x_i||^2_2/2$) has infinitely many $y_i$ that map to the same $x_i$. As a result gradient descent in the primal-space will not satisfy the Markov property. 

Instead, we build the primal space Markov chain specifically for (\ref{eqn:S.MWU}).
Let $R_{i\hat{s}_i}(s,x)= \frac{x_{i\hat{s}_i}\exp{(\eta_i\cdot  \sum_{j\neq i} e_{\hat{s}_i}A^{(ij)}x_j)}}{\sum_{\bar{s}_i\in {\cal S}_i}x_{i\bar{s}_i}\exp{(\eta_i\cdot  \sum_{j\neq i} e_{\bar{s}_i}A^{(ij)}x_j)}}$ so that $X^t=R(s,x^{t-1})$ when strategy $s$ is realized by the distribution $x^{t-1}$.
For this definition, we also set $0/0=0$. 
The probability transition kernel is then given by
\begin{align*}\label{eqn:PrimalKernel}
P(x^{t-1},X^t=x^t)= \sum_{{s \in {\cal S}:}\atop{R(s,x^{t-1})=x^t}}\prod_{i\in {\cal N}} x_{is_i}^{t-1}.\tag{Probability Transition Kernel for \ref{eqn:S.MWU}}
\end{align*}

\begin{restatable}[]{theorem}{Feller}\label{thm:Feller}
	The Markov chain $\{X^t\}_{t=0}^\infty$ updated with (\ref{eqn:PrimalKernel}) is a Feller chain when strategies are updated with  (\ref{eqn:S.MWU}) in a non-trivial 2-agent zero-sum game. 
\end{restatable}

The proof of Theorem \ref{thm:Feller} consists of standard techniques to show convergence in distribution and is in Appendix \ref{app:Feller}.

\begin{restatable}[]{theorem}{Stationary}\label{thm:Stationary}
	Let $\Gamma({\cal N}, {\cal S}, A)$ be any 2-agent zero-sum game where every element of $A=A^{(12)}$ is unique.  
	Then $\pi$ is a stationary distribution of the Feller chain $\{X^t\}_{t=0}$ created by (\ref{eqn:S.MWU}) if and only if $\pi(x)>0$ implies $x\in {\cal X}$ is a pure strategy.
\end{restatable}

The proof of Theorem \ref{thm:Stationary} mostly follows from Theorems \ref{thm:boundary} and \ref{thm:noninterior}.  
We first show that any mixture of pure strategies is a stationary distribution. 
Then, after a constructing a game on the faces of ${\cal X}$, the theorems suggest that strategies will converges the the boundaries of their respective faces implying the extreme points of ${\cal X}$ -- the pure strategies  -- are the only non-transient states. 
The full details of this proof can be found in Appendix \ref{app:stationary}.

{
\section{Simulations}
Theorem \ref{thm:Stationary} indicates that the distribution of agent strategies will be concentrated near extreme points, e.g., in a game of matching pennies, after enough iterations agent 1's strategy will be $x^t_1\approx (1,0)$ with probability $a$ and $x^t_1\approx(0,1)$ where $a+b\approx 1$. 
This indicates that the realized strategies will almost always be approximately a pure strategy. 
As depicted in Figure \ref{fig:MPstrat}, the realized strategies quickly diverge and become concentrated near the extreme points of the strategy space. 

\begin{figure}[!h]
	\centering
	\includegraphics[scale=0.3]{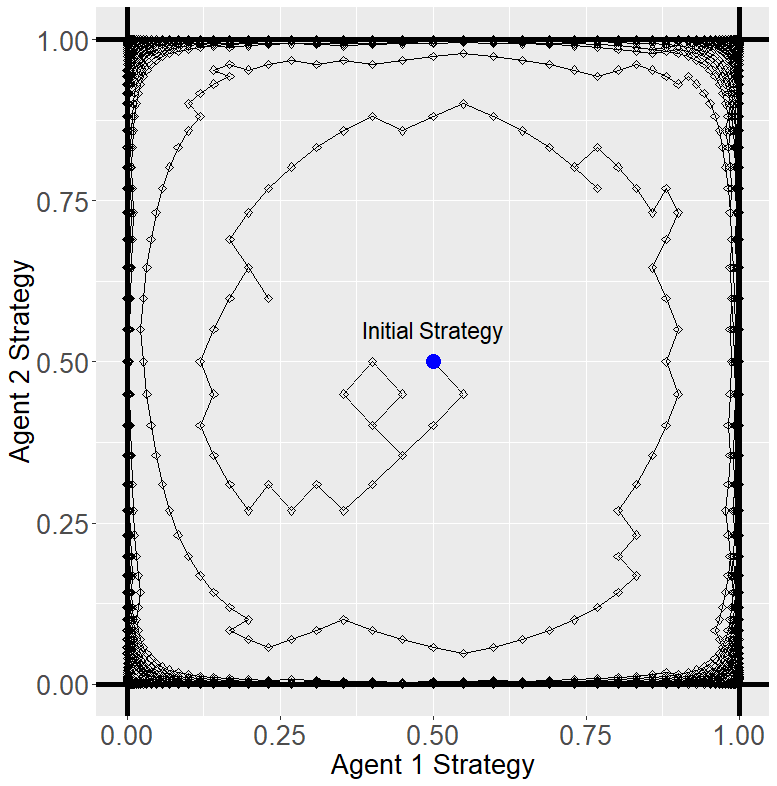}
	\caption{3000 Iterations of Stochastic MWU Applied to  Matching Pennies with $\epsilon_1=\epsilon_2=0.1$ with the Initial strategies $x_1^0=x_2^0=(0.5,0.5)$.  Despite Starting at a Nash Equilibrium, the Strategies Quickly Converge to the Boundary.}
	\label{fig:MPstrat}
\end{figure}

Interestingly, this convergence to corners actually suggests sublinear regret for the stochastic MWU algorithm with fixed learning rate:
It is well known that for MWU that agent 1's regret is bounded by $\sum_{t=0}^T \langle x_1^{t+1}-x_1^t, A x_{2}^t \rangle+O(1)$ \cite{SS11}. 
If agent 1 is stuck near the same extreme point for two consecutive iterations, then $x_1^{t+1}\approx x_1^t$ and that iteration's contribution to regret is $\langle x_1^{t+1}-x_1^t, A x_{2}^t \rangle\approx 0$, i.e., we would expect the regret to not grow in most iterations -- this specific property was exploited in \cite{BaileyNips19} to show regret grows at rate $O(\sqrt{T})$ in deterministic gradient descent with fixed learning rate in 2-agent, 2-strategy games.


Admittedly, while Theorem \ref{thm:Stationary} indicates the strategies will concentrate near extreme points, it does not directly say consecutive iterations will be near the same extreme point. 
Rather, the dynamics of MWU give us this insight:
As shown in Figure \ref{fig:ExtremePointMP10000}, the strategies go through long stretches of being within $0.1$ of the closest extreme separated by short intervals where the strategies are approximately $0.707$ away from the closest extreme point -- this corresponds to the maximum distance between the boundary and all extreme points. 
However, as shown in Figure \ref{fig:MPstrat}, the strategies are mostly moving clockwise along the boundary of the strategy space.
MWU takes time to move from one corner to another, e.g., in Figure \ref{fig:MPstrat} it consistently takes 6 iterations to move from $x_1\approx (0.375,0.625)$ to $(0.625,0.375)$.
Since, in the limit, almost all iterations are close to a pure strategy, MWU cannot switch corners often and therefore consecutive iterations are typically near the same extreme point.
As such, we expect $x^{t+1}\approx x_t$ for most iterations and that regret frequently will not grow.  
We remark that while clockwise rotations do not necessarily exist in higher dimension games, it still takes a significant number of iterations to move between extreme points and we expect for their to be long stretches of iterations where strategies are close to the same equilibrium as depicted in Figure \ref{fig:ExtremePoint10Strat} for 2-agent, 10-strategy games.

\begin{figure}[!h]
	\centering
	\def\size{0.35}
	\subfigure[Matching Pennies]{\label{fig:ExtremePointMP10000}
	    \includegraphics[scale=\size]{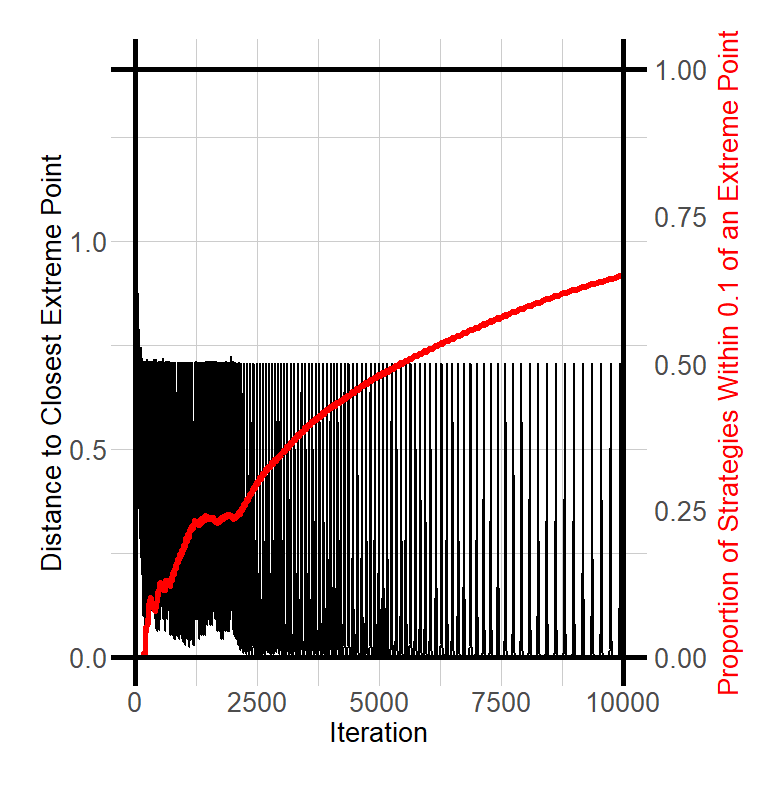}}
	\subfigure[10 Strategy Game]{\label{fig:ExtremePoint10Strat}
	    \includegraphics[scale=\size]{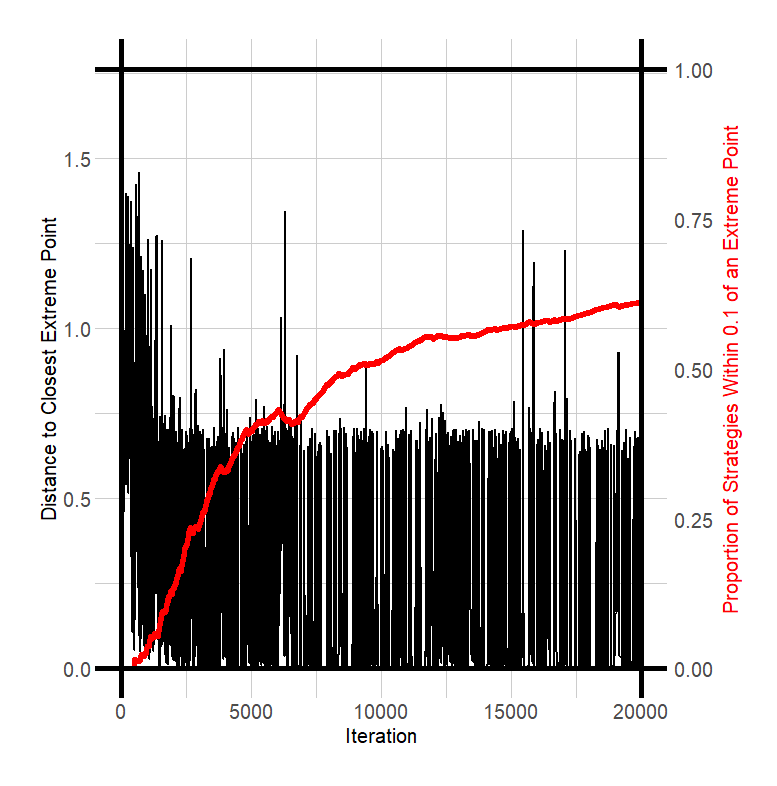}}
	    \caption{Agent's Combined Distance to the Closest Pure Strategy.}
	    \label{fig:proportion}
\end{figure}



To test this possibility of sublinear regret, we generate random games and simulate 20,000 iterations of stochastic multiplicative weights.
After storing the strategy for each iteration, we build an approximation of agent 1's regret throughout the simulation by using the model $Regret_1(t) \approx t^\alpha$.

Specifically, we estimate $\alpha$ via the linear regression $\log (R_1(t)) \approx \alpha \log (t)+\log(\beta)$.  
We repeat this process for 30 times for games with 10, 20, and 40 strategies per agent. 
Since regret tends to oscillate and is sometimes negative, we actually model $\overline{Regret}_1(T)=\max_{t\in [T]} Regret_1(t)$ -- an upper bound on agent 1's regret.  
As shown in Figure \ref{fig:regret}, the distinction is relatively small. 
R version 4.1.1 was used to complete the experiments. 
The source code is available at \href{http://jamespbailey.com/StochasticMWU/}{http://jamespbailey.com/StochasticMWU/}.

The results of the simulations are shown in Table \ref{tab:regret}.  In all 90 instances, regret appears to be growing at sublinear rate between $t^{0.3342}$ and $t^{0.7276}$ with all estimates close to $\sqrt{t}$ as shown in Figure \ref{fig:regret}. 
While individual iterates vary from this estimate -- regret tends to oscillate and will not perfectly follow the curve $\beta\cdot t^\alpha$, -- the regression captures between 87.7\% and 99.3\% of the variability for each model. 
These experiments suggest a need for more research into understanding the connection between learning dynamics and regret.

\begin{table}[!ht]
	\centering
	\caption{Estimates for the Growth of Regret ${R}_1(t)\approx t^\alpha$ and Proportion of Variability ($R^2$) Explained by Each Model.}
	\label{tab:regret}
	\begin{tabular}{| c c c|}
		\hline
		\#Strategies & Estimate of $\alpha$ & $R^2$ \\
		\hline
		10	& $0.3342-0.6694$	& $87.7\%-99.0\%$	\\
		20	& $0.3760-0.6133$	& $93.7\%-98.8\%$	\\
		40	& $0.4536-0.7276$	& $94.3\%-99.3\%$	\\
		\hline
	\end{tabular}
\end{table}

\begin{figure}[!h]
	\centering
	\includegraphics[scale=0.3]{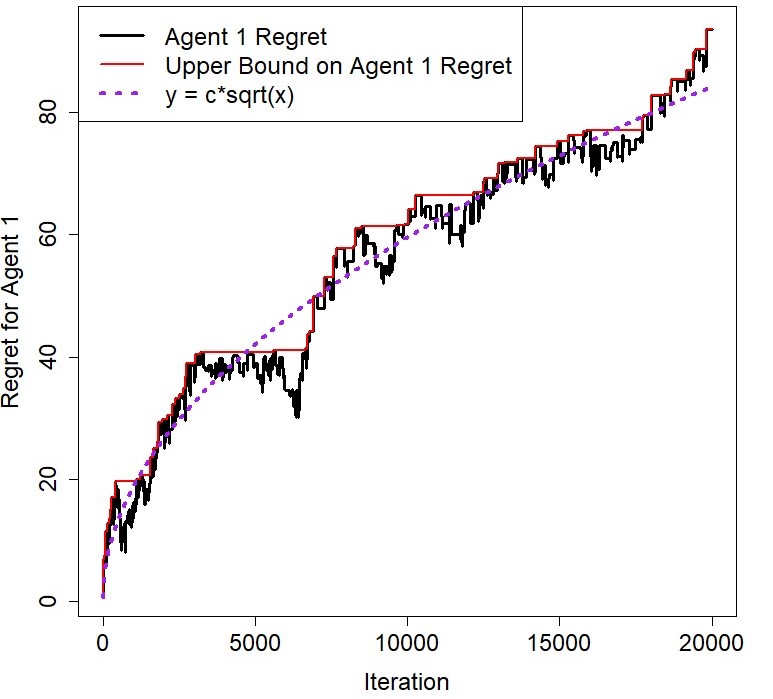}
	\caption{Approximation of Agent 1's Regret in the 1st Experiment for a Zero-Sum Game with 10 Strategies for Each Agent. As Shown by the Dashed Line, the Regret Grows Approximately at Rate $O(\sqrt{t})$.}
	\label{fig:regret}
\end{figure}

\section{Conclusion}


Our analysis of stochastic Multiplicative Weights Updates (MWU) in zero-sum games  shows that it is possible to characterize the day-to-day behavior (sometimes referred to as last-iterate behavior) of classic online algorithms even in games where they are unstable. 
%
 Our results shows that the actual realized behavior concentrates around deterministic strategy profiles, which is anthithetical to the predictions of Nash equilibrium for most prototypical zero-sum games such as Matching-Pennies or Rock-Paper-Scissors.
 Such results are clearly significantly stronger that previously known instability results or divergence to boundary results.
 
 Extending these results to other learning dynamics, games, as well as in the case of dynamically decreasing step-sizes is an interesting direction for future work.
Moreover, as we argue, these characterizations have the potential to improve optimality guarantees for the wide array of applications where MWU and FTRL dynamics are used.

}

\section{Acknowledgements}

This research/project is supported in part by the National
Research Foundation, Singapore under its AI Singapore Program (AISG Award No: AISG2-RP-2020-016), NRF2019-NRFANR095 ALIAS grant, grant PIE-SGP-AI-2018-01, NRF
2018 Fellowship NRF-NRFF2018-07 and AME Programmatic
Fund (Grant No. A20H6b0151) from the Agency for Science, Technology and Research (A*STAR).

\bibliographystyle{acm}
\bibliography{refer,ref,IEEEabrv,Bibliography}

\appendix

\section{Irreducibly of the Dual-space Markov Chain}\label{app:Irreducible}

In our proofs, we work with a slightly different dynamic where the first component of $Y_{i}^t$ is normalized to be zero for each agent. Formally, 
\begin{equation}\label{eqn:S.FTRL2}\tag{Stochastic FTRL 2}
\begin{aligned}
{Y}_{i}^t &={Y}_{i}^{t-1} +\sum_{j\neq i} A^{(ij)}s_j -\left[\sum_{j\neq i} A^{(ij)}s_j\right]_1\cdot \mathbf{1} \ \text{with  probability} \ \prod_{j\neq i} X_{js_j}^{t-1}\\
{X}_i^{t} &= \argmax_{x_i \in \mathcal{X}} \left\{\left<Y_i^{t},x_i\right>-\frac{h_i(x_i)}{\eta_i}\right\}\\
\end{aligned}
\end{equation} 
where $\mathbf{1}$ is a vector of 1's. 
We normalize the first component to clarify the dual space of payoff vectors. In (\ref{eqn:S.FTRL2}), adding a vector of constants onto $Y_i^t$ has no impact on $X_i^t$ as $Y_i^t$ is multiplied by a probability vector.  
Specifically, adding a constant vector onto $Y_i^t$ only increases $\langle Y_i^t,x_i\rangle$ by a constant and has no impact in the selection of $X_i^t$. 

In particular, this normalization makes it simpler to define the space dual to ${\cal X}_i$. 
Since ${\cal X}_i$ describes a set of probability vectors, the dimension of ${\cal X}$ is $S_i-1$ implying the dual space is also dimension $S_i-1$. However, $Y_i^t\in {\mathbb{R}}^{S_i}$ creating some ambiguity in the definition of the dual space.
However by normalizing the first component of $Y_i^t$, the dual space can be expressed simply as ${\cal Y}_i=\{y_i\in \mathbb{R}^{S_i}: y_{i1}=0\}$ without any ambiguity. 
This definition also allows us to establish a bijection between the two spaces. 

\begin{lemma}\label{lem:dual}
    If the regularlizer used in (\ref{eqn:S.FTRL2}) satisfies $\nabla h_i^*(y_i)>0$ for all $y$, then $\nabla h_i-[\nabla h_i]_1\cdot \mathbf{1}$ is a bijection between $rel.int({\cal X})$ and ${\cal Y}_i=\{y_i\in \mathbb{R}^{S_i}: y_{i1}=0\}$
\end{lemma}

\begin{proof}
	The KKT conditions of (\ref{eqn:S.FTRL2}) are given by
	\begin{align*}
		\nabla h_i(X_i^t)&= Y_i^t - \lambda_i^t\cdot \mathbf{1} + \delta_i^t\tag{Stationarity}\\
		\sum_{s_i\in {\cal S}_i} X_{is_i}^t&=1 \tag{Primal Feasibility 1}\\
		X_{is_i}^t&\geq 0 \ \forall s_i \in {\cal S}_i \tag{Primal Feasibility 2}\\
		\delta^t_{is_i} &\geq 0 \ \forall s_i \in {\cal S}_i \tag{Dual Feasibility}\\
		\langle \delta_i^t, X_i^t\rangle &= 0 \tag{Complimentary Slackness}\label{eqn:comp}
	\end{align*}
	where $\mathbf{1}$ is a vector of 1's, $\lambda_i^t$ is a dual variable associated with the constraint $\sum_{s_i\in {\cal S}_i} X_{is_i}^t=1$ and 
	$\delta^t_{is_i}$ is a dual variable associated with the constraint $X_{is_i}^t\geq 0$. 
	Since $X_i^t=\nabla h_i^*(Y_i^t)$ in the interior of ${\cal X}_i$, $X_{is_i}^t>0$ and $\delta_{is_i}=0$ by (\ref{eqn:comp}).
	Thus, $\nabla h_i(X_i^t)+ \lambda_i^t\cdot \mathbf{1}= Y_i^t $. 
	Moreover, since $Y_{i1}^t=0$, $\lambda _i^t=-[\nabla h_i(X_i^t)]$ and therefore $Y_i^t=\nabla h_i(X_i^t)-[\nabla h_i(X_i^t)]_1:=g_i(X_i^t)$. 
	Therefore $g_i$ is injective. 
	
	To see that $g_i$ is surjective, first observe that ${\cal X}_i$ is compact and therefore (\ref{eqn:S.FTRL2}) is well defined for all $Y_i\in {\cal Y}_i$ and therefore there always exists a $X_i\in {\cal X}$ such that $g(X_i)=Y_i$. 
	It remains to show that $g_i(X_i)=g_i(X'_i)$ implies $X_i=X'_i$. 
	For contradiction, suppose $g_i(X_i)=g_i(X'_i)$ for some $X_i\neq X'_i$. 
	In definition of (\ref{eqn:S.FTRL2}), $h_i$ is strictly convex and therefore so is $g_i$. 
	Thus,
	\begin{align*}g_i(X_i)&> g_i(X_i') + \langle \nabla g_i(X_i'), X_i-X'_i\rangle\\
	&> g_i(X_i) + \langle \nabla g_i(X_i), X'_i-X_i\rangle\\
	&\phantom{>}+ \langle \nabla g_i(X_i'), X_i-X'_i\rangle\\
	&=g_i(X_i)
	\end{align*}
	since $\nabla g_i(X_i)=\nabla g_i(X'_i)$ and $g_i(X_i)>g_i(X_i)$, a contradiction.
	Therefore $X_i=X'_i$ and $g_i$ is a bijection between ${\cal Y}_i$ and ${\cal X}_i$. 
\end{proof}

With this definition of (\ref{eqn:S.FTRL2}), we  update our Markov Chain accordingly:

\begin{equation}\tag{States Reachable After $t$ Iterations}
\begin{aligned}
{\cal Y}^0&= y_i^0 \\
{\cal Y}^t&= \bigcup_{y\in {\cal Y}^{t-1}} \bigcup_{{s\in {\cal S}}} \left(y_1 + \sum_{j\neq 1}A^{(1j)}e_{s_j}-\left[\sum_{j\neq 1}A^{(1j)}e_{s_j}\right]_1,\right. 
\left. \cdots, y_{N} + \sum_{j\neq {N}}A^{({N}j)}e_{s_j} -\left[\sum_{j\neq 1}A^{(Nj)}e_{s_j}\right]_1\right)
\end{aligned}
\end{equation}

\begin{align*}
\bar{P}(y^{t-1},Y^t=y^t)=\sum_{{s\in S:} \atop {y_i^t ={y}_{i}^{t-1} +\sum_{j\neq i} A^{(ij)}s_j-\left[\sum_{j\neq 1}A^{(ij)}e_{s_j}\right]_1} \ \forall_{i=1}^N} \prod_{j=1}^N x_{js_j}^{t-1}\tag{Transition Kernel in Dual-space}
\end{align*}
where ${x}_i^{t-1} = \argmax_{x_i \in \mathcal{X}} \left\{\left<y_i^{t-1},x_i\right>-\frac{h_i(x_i)}{\eta_i}\right\}$ is the realization of $X_i^{t-1}$ as given in the definition of (\ref{eqn:S.FTRL2}).

\Irreducible*

We remark that the proof of this theorem also holds for the Markov chain defined in the main body using (\ref{eqn:S.FTRL}) following the same steps. 
However, we formally show the result for the Markov chain created by (\ref{eqn:S.FTRL2}) as the duality established in Lemma \ref{lem:dual} will be important later. 

\begin{proof}
	Irreducibly requires that every state can be reached by any other state with positive probability.  It suffices to that every state can be reached from $y^0$ and then can return to $y^0$. 
	
	First, observe that every state in ${\cal Y}^t$ is reachable in $t$ steps from $y^0$ with positive probability:
	We proceed by induction and the base case $(t=0)$ holds trivially. 
	By definition of ${\cal Y}^t$, for any $y^t\in {\cal Y}^t$, there is a state $y^{t-1}\in {\cal Y}^{t-1}$ where $y^t$ is reached from $y^{t-1}$ via pure strategy $s$. 
	Since $x_i^{t-1}=\nabla h_i^*(y_i^{t-1})>0$ for all $i$, 
	By the inductive hypothesis, there is also positive probability to go from state $y^0$ to $y^{t-1}$ in $t-1$ steps. 
	Thus, there is positive probability of going from state $y^0$ to $y^{t-1}$ to $y^t$ in $t$ steps.
	
	Next, we show that for each state $y^t\in {\cal Y}^t$ that we can return to $y^0$. 
	Since $x^*$ is rational, there exists a $b\in \mathbb{Z}_{>0}$ and $c_{is_i}\in \mathbb{Z}_{>0}$ such that $x_{is_i}^*=c_{is_i}/b$ for all $s_i \in {\cal S}_i$ and $i\in {\cal N}$.
	We now show that $y^0$ can be reached in $(b-1)\cdot t$ iterations. 
	
	Let $y\in {\cal Y}^{t-1}$ and $\bar{s}\in {\cal S}$ be such that $y_i^t=y_i + \sum_{j\neq i} A^{(ij)}e_{\bar{s}_j}-\left[  \sum_{j\neq i} A^{(ij)}e_{\bar{s}_j}\right]_1$ for all $i\in {\cal N}$.  Such a pair exists by definition of ${\cal Y}^t$.

	Now consider any sample path of (\ref{eqn:S.FTRL}) where agent $i$ plays strategy $\bar{s}_i$ a total of $(c_{i\bar{s}_i}-1)$ times, and strategy $s_i$ a total of $c_{is_i}$ times for all $s_i\in {\cal S}_i\setminus\{\bar{s}_i\}$ over the next $b-1$ iterations.  Since $\nabla h_i^*(y_i)>0$ for all $y_i \in {\cal Y}_i$, there is positive probability of this path occurring. 
	Moreover, these updates return the state of the system to $y$:
	
	\begin{align*}
	y^{t+b-1}_i	&=y^t_i + \sum_{j\neq i}\left((c_{j \bar{s}_j}-1)\cdot A^{(ij)}e_{\bar{s}_j} + \sum_{s_i\in S_i\setminus\{\bar{s}_i\}}c_{j s_j}\cdot A^{(ij)}e_{s_j}  \right)\\
	& \phantom{=y_i^t} -\sum_{j\neq i}\left[\left((c_{j \bar{s}_j}-1)\cdot A^{(ij)}e_{\bar{s}_j} + \sum_{s_i\in S_i\setminus\{\bar{s}_i\}}c_{j s_j}\cdot A^{(ij)}e_{s_j}  \right)\right]_1\\
	&=y_i+ \sum_{j\neq i} A^{(ij)}e_{\bar{s}_j} + \sum_{j\neq i}\left((c_{j \bar{s}_j}-1)\cdot A^{(ij)}e_{\bar{s}_j} + \sum_{s_i\in S_i\setminus\{\bar{s}_i\}}c_{j s_j}\cdot A^{(ij)}e_{s_j}  \right) \\
	& \phantom{=y_i}-\left[\sum_{j\neq i} A^{(ij)}e_{\bar{s}_j} \right]_1 -\sum_{j\neq i}\left[\left((c_{j \bar{s}_j}-1)\cdot A^{(ij)}e_{\bar{s}_j} + \sum_{s_i\in S_i\setminus\{\bar{s}_i\}}c_{j s_j}\cdot A^{(ij)}e_{s_j}  \right)\right]_1\\
	&=y_i+ \sum_{j\neq i}\sum_{s_i\in S_i}c_{j s_j}\cdot A^{(ij)}e_{s_j}  -\sum_{j\neq i}\left[\sum_{s_i\in S_i}c_{j s_j}\cdot A^{(ij)}e_{s_j}  \right]_1\\
	&=y_i+ b\cdot \sum_{j\neq i}\cdot A^{(ij)}x^*_j  -b\cdot \sum_{j\neq i}\left[A^{(ij)}x^*_j  \right]_1=y_i 
	\end{align*}
	
	since  $\sum_{j \neq i} A^{(ij)}x^*_j= {\bf 0}$.
	
	This holds for each agent $i$ and the state $y$ can be reached from $y^t$ after $b-1$ iterations. 
	Inductively, this implies that $y^0$ can reached from $y^t$ after $t\cdot (b-1)$ iterations. 
	Thus every state can be reached from $y^0$ and then return to $y^0$ and the Markov chain is irreducible. 

\end{proof}

\section{Convergence for Interior Nash}\label{app:BoundaryProofs}

In this section, we provide the proofs of Theorems \ref{thm:returntime} and \ref{thm:boundary}.
Specifically, we make use of the Markov chain for the dual space ${\cal Y}$ given in Section \ref{sec:dual}. 
Prior to making use of the Markov chains, we show that, in expectations, (\ref{eqn:D.FTRL}) will not move toward a Nash equilibrium $x^*$. 

\subsection{Positive Drift}

Prior to showing (\ref{eqn:S.FTRL}) causes strategies to drift away from a Nash equilibrium, we first show a result relating the Bregman divergence and Fenchel-coupling.

\begin{lemma}\label{lem:FBequiv}
	Suppose $\eta_i=1$ and $x^t$ and $y^t$ are obtained from (\ref{eqn:D.FTRL}) and be such that $x^t$ is in the interior of ${\cal X}$. 
	If $x^*$ is a fully mixed Nash equilibrium, then $D_h(x^*||x^t)=F_h(x^*||y^t)$. 
\end{lemma}

\begin{proof}
	Following identically to the proof of Lemma \ref{lem:dual}, the KKT condition of (\ref{eqn:D.FTRL}) imply that $\nabla h_i(x_i^t)= y_i^t - \lambda_i^t\cdot \mathbf{1}$. 
	
	Next, recall $h_i^*(y_i^t)=\sup_{x_i \in {\cal X}_i} \left\{\langle y_i^t,x_i\rangle -h_i(x_i) \right\}$. Since $x_i^t$ is selected by (\ref{eqn:D.FTRL}), $h^*_i(y_i^t)= \langle y_i^t,x_i^t\rangle -h_i(x_i^t)$.  
	
	Finally, the Bregman divergence is given by
	\begin{align*}
		D_h(x^*||x^t)=&
		\sum_{i \in {\cal N}} \left( h_i(x_i^*)-h_i(x_i^t)-\langle \nabla h_i(x_i^t), x_i^*-x_i^t \right)\\
		=&\sum_{i \in {\cal N}} \left( h_i(x_i^*)-h_i(x_i^t)-\langle \nabla y_i^t - \lambda_i^t\cdot \mathbf{1}, x_i^*-x_i^t \rangle\right)\\
		=&\sum_{i \in {\cal N}} \left( h_i(x_i^*)+\langle y_i^t,x_i^t\rangle-h_i(x_i^t)-\langle  y_i^t , x_i^*\rangle \right)\\
		=&\sum_{i \in {\cal N}} \left( h_i(x_i^*)h_i^*(y_i^t)-\langle  y_i^t , x_i^*\rangle \right)\\
		=&F_h(x^*||y^t)
	\end{align*}
	since $\lambda_i^t \langle \mathbf{1}, x_i^*-x_i^t \rangle =0$ as both $x_i^*$ and $x_i^t$ are probability vectors. 
	Thus, the Fenchel-coupling and the Bregman divergence are equivalent. 
\end{proof}

With this equivalence, we show that the Fenchel-coupling is increasing. 

\begin{theorem}\label{thm:Drift}
	The Fenchel-coupling is increasing in expectations, i.e., $E[F_h(x^*||Y^t)|Y^{t-1}=y^{t-1}] > F_h(x^*||y^{t-1})$. 
\end{theorem}

\begin{proof}
	Let $y_i^t= E[Y_i^t| Y_i^{t-1}=y_i^{t-1}]$.  
	As shown in Section \ref{sec:Stochastic}, $y_i^t$ is also the update of $y_i^{t-1}$ in (\ref{eqn:D.FTRL}). 
	Since the convex-conjugate $h^*_i$ is convex, the Fenchel-coupling $F(x^*||y)=h_i^*(y_i)-\langle y_i, x_i^*\rangle +h_i(x_i)$ is convex with respect to $y$.  
	Thus, by Jensen's inequality, 
	\begin{align*}
		E[F_h(x^*||Y^t)|Y^{t-1}=y^{t-1}]
		\geq F_h(x^*||E[Y^t|Y^{t-1}=y^{t-1})]= F_h(x^*||y^t)
		\end{align*}
	I.e., after one iteration, the expected Fenchel-coupling of (\ref{eqn:S.FTRL}) is at least as large as the Fenchel-coupling of (\ref{eqn:D.FTRL}). 
	
	By Lemma \ref{lem:FBequiv}, $F_h(x^*||y^t)= D_h(x^*||y^t)$ and $F_h(x^*||y^{t-1})= D_h(x^*||y^{t-1})$.
	Finally, by Theorem 4.1 of \cite{BaileyEC18}, $D_h(x^*||y^t)> D_h(x^*||y^{t-1})$. 
	Combining all parts completes the proof of the Theorem, i.e., 
	\begin{align*}
	E[F_h(x^*||Y^t)|Y^{t-1}=y^{t-1}]
	\geq  F_h(x^*||y^t)=D_h(x^*||y^t)
	> D_h(x^*||y^{t-1})=F_h(x^*||y^{t-1}).
	\end{align*}
\end{proof}

\subsection{Proof of Theorem \ref{thm:returntime}}

Let $\tau_{\bar{B}}^+$ be the time to return to any state $y\in \bar{B}\subseteq {\cal Y}$ after leaving the set $\bar{B}$.
The term $\tau_{\bar{B}}^+$ is known as a \emph{stopping time} and we make use of a Corollary from \cite{menshikov2016} to show Theorem \ref{thm:returntime}.

\begin{lemma}[see \cite{menshikov2016},  Corollary 2.6.11]\label{thm:Null}
	Let $\{Y^t\}_{t=0}^\infty$ be an irreducible Markov chain on a countable state space ${\cal Y}$.  Suppose there exists a non-empty $\bar{B} \subseteq {\cal Y}$ and a function $f: {\cal Y}\to \mathbb{R}_+$ such that $f(Y^T)$ is integrable, and 
	\begin{align}
	E[f(Y^{t+1})-f(Y^t)||Y^t=y] \geq 0 &\text{ for all } y\in {\cal Y}\setminus \bar{B};\label{eqn:cond1}\\
	E[(f(Y^{t+1})-f(Y^t))^+||Y^t=y] \leq b \label{eqn:cond2} &\text{ for all } y\in {\cal Y}\setminus \bar{B};\\
	f(y)> \max_{z\in \bar{B}} f(z) &\text{ for some } y\in {\cal Y}\setminus \bar{B}.\label{eqn:cond3}
	\end{align}
	Then $E[\tau_{\bar{B}}^+]=\infty$. 
\end{lemma}

Using this result, we now prove that the expected time to return to any set when strategies are updated with (\ref{eqn:S.MWU}) is infinity. 
We prove the result for any variant of (\ref{eqn:S.FTRL}) where $x_i=\nabla h^*_i(y_i)>0$ for all $y_i\in \mathbb{R}^{S_i}$ and for all $i$, i.e., for any update rules that guarantees players will always select fully mixed strategies. 

\Return*

\begin{proof}
    Recall from Lemma \ref{lem:dual} $g_i(x_i)=\nabla h_i(x_i)-[\nabla h_i(x_i)]_1$ is bijection between the relative interior of the primal space ${\cal X}$ and the dual space ${\cal Y}$. 
    For a compact $B\subset rel.int({\cal X})$, let $\bar{B}=\bigcup{x\in X}\{g_1(x_1),...,g_N(x_N)\}$ be the corresponds set of payoff vectors that yield $B$ according to (\ref{eqn:S.FTRL2}).  
    Since $B$ is compact and $g_i$ is continuous, $\bar{B}$ is also compact.  
    Thus, to show that we expect an infinite amount of time for $\{X\}_{t=0}^T$ to return to ${B}$, it suffices to show that we expect an infinite amount of time for $\{Y\}_{t=0}^T$ to return to $\bar{B}$.

	The statement of Theorem \ref{thm:boundary} allows for arbitrarily values of $\langle x^*_i, \sum_{j\neq i} A^{(ij)} x_j^* \rangle$ while our Markov chain is specifically built for games where $\langle x^*_i, \sum_{j\neq i} A^{(ij)} x_j^* \rangle=0$.  
	First observe that (\ref{eqn:S.FTRL}) is invariant to constant shifts in $A^{(ij)}$ and therefore without loss of generality we may assume $\langle x^*_i, \sum_{j\neq i} A^{(ij)} x_j^* \rangle={\bf 0}$ for all $i\in {\cal N}$.
	
	Next, Let $\bar{B}$ be any compact set in ${\cal Y}$.
	We now show that the expected time to return to any state in $\bar{B}$ is infinity when strategies are updated via (\ref{eqn:S.MWU}). 
	It suffices to select a function  $f:{\cal Y}\to \mathbb{R}_+$ such that $\bar{B}, {\cal Y}, f$, and $\{Y^t\}_{t=0}^\infty$ satisfy the conditions of Theorem \ref{thm:Null}.
	
	We select $f=F_h$ (\ref{eqn:Fenchel}). Equivalently, this is the Bregman divergence of $X^t=\nabla h^*_i(Y_i^t)$ when $X^t$ is in the interior as is the case for (\ref{eqn:S.MWU}). 
	We now show our selection of $f$ satisfies the conditions of Theorem \ref{thm:Null}.  
	
	In Theorem \ref{thm:Irreducible}, we establish that $\{{Y}^t\}_{t=0}^\infty$ is an irreducible Markov chain on a countable state space if $\nabla h^*_i(y_i)>0$ for all $i$. 
	Moreover, $F_h(x^*||Y^t)$ is trivially integrable;  $Y^t$ comes from the finite set ${\cal Y}^t$ and therefore $E[Y^t]<\infty$ since $F_h$ maps ${\cal Y}$ to $\mathbb{R}_+$. 
	It remains to show the conditions (\ref{eqn:cond1})-(\ref{eqn:cond3}).
	{\color{black}
	(\ref{eqn:cond1}) holds by Lemma \ref{thm:Drift}. }

	For condition (\ref{eqn:cond3}), let $\bar{y}=\argmax_{y \in \bar{B}} F_h(x^*||y)$.  Such a $\bar{y}$ exists since $\bar{B}$ is compact and $F_h$ is continuous. 
	By Lemma \ref{lem:FBequiv}, if $\bar{y}$ is updated with (\ref{eqn:S.FTRL}), then the Fenchel-coupling is expected to increasing implying there is a $z\in {\cal Y}$ such that $F_h(x^*||z)> F_h(x^*||\bar{y})$.  Moreover, $z\in {\cal Y} \setminus B$ by selection of $\bar{y}$ and $z$ satisfies condition (\ref{eqn:cond3}).

	It remains to show (\ref{eqn:cond2}). 
	We show a stronger condition.  
	We show there exists a $b$ such that
	\begin{align*}\sum_{i\in {\cal N}}F_{h_i}(x^*_i||y_i+\sum_{j\neq i} A^{(ij)}e_{s_j}) - \sum_{i\in {\cal N}}F_{h_i}(x^*_i||y_i) \leq b
	\forall \ s\in {\cal S} \text{ and } y\in {\cal Y},\end{align*} 
	i.e., that $F_h$ increases by at most $b$ regardless of the current location and the sample path.  
	This certainly implies that the expectation is bounded when $y^T\notin \bar{B}$. 
	
	Following identically to the proof of Lemma \ref{thm:Drift}, 
	\begin{align*}\sum_{i\in {\cal N}}F_{h_i}(x^*_i||y_i+\sum_{j\neq i} A^{(ij)}e_{s_j}) - \sum_{i\in {\cal N}}F_{h_i}(x^*_i||y_i) 
	=\sum_{i\in {\cal N}}\left( h^*_i(y_i+\sum_{j\neq i} A^{(ij)}e_{s_j})-h^*_i(y_i)\right)\end{align*}
	
	It is well known that the convex conjugate $h^*_i$ is a convex function and therefore
	\begin{align*}
	h^*_i(y_i)\geq h^*_i(y_i+\sum_{j\neq i} A^{(ij)}e_{s_j})  -\langle \nabla h^*_i(y_i+\sum_{j\neq i} A^{(ij)}e_{s_j}), \sum_{j\neq i} A^{(ij)}e_{s_j} \rangle .
	\end{align*}
	Thus, 
	\begin{align*}\sum_{i\in {\cal N}}F_{h_i}(x^*_i||y_i+\sum_{j\neq i} A^{(ij)}e_{s_j}) - \sum_{i\in {\cal N}}F_{h_i}(x^*_i||y_i) &=\sum_{i\in {\cal N}}\left( h^*_i(y_i+\sum_{j\neq i} A^{(ij)}e_{s_j})-h^*_i(y_i)\right)\\
	&\leq \sum_{i \in {\cal N}}\langle \nabla h^*_i(y_i+\sum_{j\neq i} A^{(ij)}e_{s_j}), \sum_{j\neq i} A^{(ij)}e_{s_j} \rangle\\
	&\leq \sum_{\in {\cal N}}\max_{x_i\in {\cal X}_i} \langle x_i, A^{(ij)}e_{s_j}\rangle
	\end{align*}
	where (since) $\nabla h_i^*$ maps $y_i$ to the set of probability vectors (${\cal X}_i$). 
	Therefore, (\ref{eqn:cond2}) holds with $b=\sum_{\in {\cal N}}\max_{x_i\in {\cal X}_i} \langle x_i, A^{(ij)}e_{s_j}\rangle$ which is finite since ${\cal X}_i$ is compact and $\langle x_i, A^{(ij)}e_{s_j}\rangle$ is linear with respect to $x_i$. 
	The conditions of Theorem \ref{thm:Null} are satisfied and thus we expect an infinite number of iterations to pass before the strategies return to the set $\bar{B}$.
\end{proof}

\subsection{Proof of Theorem \ref{thm:boundary}}

Next, we show that the proportion of iterations $Y^t\in \bar{B}$ for any compact $\bar{B}$ is 0.  Define the number of times up to time $T$ where $Y^t=y\in {\cal Y}$ as
\begin{align}
V_y(T)= \sum_{k=0}^{T-1} \mathbf{1}_{\{Y^t=y\}}.
\end{align}
The proportion of time $Y^t=y$ is then simply $V_y(T)/T$.  
By the following theorem from \cite{norris1998markov}, this proportion goes to zero almost surely. 

\begin{theorem}[See \cite{norris1998markov} Theorem 1.10.2]\label{thm:Ergodic}	For any irreducible Markov chain,
	\begin{align}
	p\left( \lim_{T\to \infty}\frac{V_y(T)}{T}\to \frac{1}{E[\tau_y^+]}\right)=1.
	\end{align}
\end{theorem}

With this result, the proof of Theorem \ref{thm:boundary} is straightforward.  

\Boundary*

\begin{proof}
	Similar to the previous theorem, it suffices to show the result for $\langle x^*_i, \sum_{j\neq i} A^{(ij)} x_j^* \rangle=0$.
    It also suffices to show the result for $\{Y^t\}_{t=0}^T$ and $\bar{B}$. 
    The proportion of the time that the strategies appear in $\bar{B}$ is given by 
    \begin{align}
	\sum_{y\in {\cal Y}\cap \hat{B}} \frac{V_y(T)}{T}
	\end{align}
    
    In this proof, we actually assume $A$ is rational, which guarantees a rational Nash equilibrium. 
    Since $A^{(ij)}$ is rational, $\sum_{j\neq i} A^{(ij)}e_{s_j}$ can be expressed as $v_i(s_j)/b$ where $b\in \mathbb{Z}$ and $v_i(s_j)\in \mathbb{Z}^{S_i}$ for $s_i\in S_i$ and $i\in {\cal N}$, i.e., the updates to $y_i$ have a common denominator.
    Thus, every $y{\cal Y}$ can be expressed as $y^0+z/b$ for some $z\in \mathbb{Z}^{\prod_{i\in {\cal N} S_i}}$. 
    Since every $y\in {\cal Y}$ can be expressed with the same denominator, there are finitely many $y\in \hat{B}$.  I.e., there exists a constant $c$ such that ${\cal Y}\cap \hat{B}=c$.

	For any $\epsilon>0$, the probability that the proportion of iterates where $Y^t\in \hat{B}$ is given by:
	\begin{align*}
	p\left( \lim_{T\to \infty}\sum_{y\in {\cal Y}\cap \hat{B}} \frac{V_y(T)}{T}>\epsilon \right)
	=\ &p\left( \sum_{y\in {\cal Y}\cap \hat{B}}\lim_{T\to \infty} \frac{V_y(T)}{T}>\epsilon \right)\\
	\leq\ & p\left( \bigcup_{y\in {\cal Y}\cap \hat{B}}\left\{\lim_{T\to \infty} \frac{V_y(T)}{T}>\frac{\epsilon}{c}\right\} \right)\\
	=\ &\sum_{y\in {\cal Y}\cap \hat{B}}p\left( \lim_{T\to \infty} \frac{V_y(T)}{T}>\frac{\epsilon}{c} \right)=0
	\end{align*}
	by the statement of Theorem \ref{thm:Ergodic}. 
	Therefore, the proportion of the time that the strategies appear in $\hat{B}$ goes to 0 almost surely.
\end{proof}

\section{Convergence for Non-Interior Nash}\label{app:Noninterior}

We begin by giving another formulation of (\ref{eqn:S.MWU}) and show it is equivalent to both (\ref{eqn:S.MWU}) and (\ref{eqn:S.FTRL}) with $h_i(x_i)= \sum_{s_i=1}^{S_i}x_{is_i}\ln x_{is_i}$.  This new form will be simpler for showing convergence to the boundary when there is not an interior Nash equilibrium. 

\begin{lemma}\label{lem:equiv}
	(\ref{eqn:S.FTRL}) with $h_i(x_i)= \sum_{s_i=1}^{S_i}x_{is_i}\ln x_{is_i}$ and (\ref{eqn:S.MWU}) are both equivalent to 
	\begin{align*}\label{eqn:S.MWU2}
	{X}_{is_i}^{t} &= \frac{\exp{(\eta_i Y_{is_i}^t)}}{\sum_{\bar{s}_i\in {\cal S}_i}\exp{(\eta_i Y_{i\bar{s}_i}^t)}}.\tag{Stochastic MWU 2}
	\end{align*}
\end{lemma}

\begin{proof}
	First, we show that (\ref{eqn:S.FTRL}) is equivalent to (\ref{eqn:S.MWU2}).  
	Recall (\ref{eqn:S.FTRL}) is
	\begin{align*}
	{X}_i^{t} &= \argmax_{x_i \in \mathcal{X}} \left\{\left<Y_i^{t},x_i\right>-\frac{h_i(x_i)}{\eta_i}\right\}\\
	&=\argmax_{x_i \in \mathcal{X}} \left\{\left<Y_i^{t},x_i\right>-\frac{\sum_{s_i=1}^{S_i}x_{is_i}\ln x_{is_i}}{\eta_i}\right\}
	\end{align*}
	Perform the variable solution $x_{iS_i}=1-\sum_{s_i=1}^{S_i-1} x_{is_i}$ and $X_i^t$ is the maximizer of
	\begin{align*}
	f(x_i)&= \sum_{s_i=1}^{S_i-1} Y_{is_i}^t \cdot x_{is_i}
	\\&\phantom{=}-\sum_{s_i=1}^{S_i-1} \frac{x_{is_i}\ln x_{is_i}}{\eta_i}+Y_{iS_i}^t \cdot (1-\sum_{s_i=1}^{S_i-1} x_{is_i})
	\\&\phantom{=}- \frac{ (1-\sum_{s_i=1}^{S_i-1} x_{is_i})\ln  (1-\sum_{s_i=1}^{S_i-1} x_{is_i})}{\eta_i}
	\end{align*}
	where the domain of $f$ is given by $\{x_i\in {\mathbb{R}^{S_i-1}}: x_{is_i}\geq 0 \ \forall s_i=1...S_i-1, \sum_{s_i=1}^{S_i-1}x_{is_i}\leq 1\}$. The function $f$ is strictly convex, and, if the optimizer is in the interior of the domain, then its optimality conditions are given by:
	\begin{align*}
	\frac{\partial f}{\partial x_{is_i}}=Y_{is_i}^t-\frac{\ln x_{is_i}}{\eta_i} - Y_{iS_i}^t + \frac{\ln (1-\sum_{\bar{s}_i=1}^{n-1}x_{i\bar{s}_i})}{\eta_i}=0 
	\end{align*} 
	Recalling $x_{iS_i}=1-\sum_{s_i=1}^{S_i-1} x_{is_i}$, these optimality conditions are rewritten as 
	\begin{align*}\exp(\eta(Y_{is_i}^t-Y_{iS_i}^t))=\frac{X_{is_i}^t}{X_{iS_i}^t} \ \forall s_i=1,...,S_i-1
	\end{align*}
	(\ref{eqn:S.MWU2}) satisfies these conditions and in the interior of the domain of $f$ and therefore (\ref{eqn:S.FTRL}) is equivalent to (\ref{eqn:S.MWU2}). 
	Next, we show that (\ref{eqn:S.MWU}) is equivalent to (\ref{eqn:S.MWU2}).
	We proceed by induction. 
	In the definition of (\ref{eqn:S.MWU}), there is no $y_i^0$ and therefore we simply select $y_i^0$ such that the result holds for $t=0$.

	Next, by definition of $Y_i^t$ and by the inductive hypothesis, observe that
	\begin{align*}
	X_{is_i}^{t-1}\exp{(\eta_i\cdot  \sum_{j\neq i} e_{s_i}A^{(ij)}X_j^{t-1})}
	=\ &X_{is_i}^{t-1}\exp{(\eta_i\cdot (Y_{is_i}^t- Y_{is_i}^{t-1}}))\\
	=\ & \frac{\exp{(\eta_i Y_{is_i}^{t-1})}}{\sum_{\bar{s}_i\in {\cal S}_i}\exp{(\eta_i Y_{i\bar{s}_i}^{t-1})}}\exp{(\eta_i\cdot (Y_{is_i}^t- Y_{is_i}^{t-1}}))\\
	=\ & \frac{\exp{(\eta_i\cdot Y_{is_i}^t})}{\sum_{\bar{s}_i\in {\cal S}_i}\exp{(\eta_i Y_{i\bar{s}_i}^{t-1})}}.
	\end{align*}	
	
	Finally, recall (\ref{eqn:S.MWU}) is given by	
	\begin{align*}
	X_{is_i}^t&=\frac{X_{is_i}^{t-1}\exp{(\eta_i\cdot  \sum_{j\neq i} e_{s_i}A^{(ij)}X_j^{t-1})}}{\sum_{\bar{s}_i\in {\cal S}_i}X_{i\bar{s}_i}^{t-1}\exp{(\eta_i\cdot  \sum_{j\neq i} e_{\bar{s}_i}A^{(ij)}X_j^{t-1})}}\\
	&= \frac{\exp{(\eta_i Y_{is_i}^t)}}{\sum_{\bar{s}_i\in {\cal S}_i}\exp{(\eta_i Y_{i\bar{s}_i}^t)}}
	\end{align*}
	by the previous observation. 
	Thus, (\ref{eqn:S.FTRL}), (\ref{eqn:S.MWU}), and (\ref{eqn:S.MWU2}) are all equivalent. 
\end{proof}

Our proof of Theorem \ref{thm:noninterior} now follows similarly to proofs from \cite{piliouras2014optimization,GeorgiosSODA18,BaileyEC18} for deterministic variants of FTRL. 
Specifically, we show that if there are \emph{non-essential} strategies (not used in a Nash equilibrium), then there is at least one where probability of playing that strategy will go to zero. 

\begin{definition}
 A strategy $s_i \in \mathcal{S}_i$ is essential iff there is a Nash equilibrium $x^*$ with $x^*_{s_i} > 0$.	
\end{definition}

\Noninterior*

\begin{proof}
\cite[Lemma C.3]{GeorgiosSODA18} shows that there if there are non-essential strategies, then there is at least agent (without of generality, agent 1) and one non-essential strategy $s_1$ such that $\langle e_{s_1}, A^{(12)}x_2^*\rangle < \langle x_1^*, A^{(12)}x_2^*\rangle$, i.e., agent 1 is strictly worse off switching from the Nash-equilibrium to the non-essential strategy $s_1$. 
Let $s_1$ be this non-essential strategy and let $s'_1$ be any essential strategy.

Without loss of generality, suppose as usual that $\langle x_1^*, A^{(12)}x_2^*\rangle =0$ and let $\delta=\langle e_{s_1}, Ax_2^*\rangle $.  By selection of $s_1$, $\delta <0$.
Next, observe that $\langle e_{s'_1}, A^{(12)}x_2^*\rangle =0$ for any essential $s'_1$ since $x^*_{1s'_1}>0$. 
By continuity of the agent's payoff function, $\langle x_1, A^{(12)}x_2\rangle$, there is a neighborhood $B$ around $x^*$ such that $\langle e_{s_1}, A^{(12)}x_2\rangle<\frac{2\delta}{3}$ and $\langle e_{s'_1}, A^{(12)}x_2\rangle>\frac{\delta}{3}$ for all $x\in B$.

Let $\epsilon =\exp(\eta)-1$.  
It is well known that there exists an $\epsilon_0$ such that for all $\epsilon<\epsilon_0$ that (\ref{eqn:S.MWU}) converges to the set of $O(\epsilon)$-Nash equilibria with probability 1 (See e.g., \cite{Cesa06}). 
Equivalently, for all $\eta< \eta_0= \ln (\epsilon_0+1)$, (\ref{eqn:S.MWU}) converges to the set of $O(\exp(\eta)-1)$-Nash equilibria. 
Formally, $x$ is an $\epsilon$-Nash equilibrium if $\langle x_i, \sum_{j\neq i} A^{(ij)} x_j\rangle\geq \langle e_{is_i}, \sum_{j\neq i} A^{(ij)} x_j\rangle-\epsilon$ for all $s_i\in {\cal S}_i$ and all $i\in {\cal N}$, i.e., deviating from $x_i$ causes agent $i$ to lose at most $\epsilon$ utility. 
By continuity of the payout function, the set of $\epsilon$-Nash equilibria converges to the set of Nash equilibria as $\epsilon\to 0$. 
As such, we select $\epsilon<\epsilon_0$ small enough such that the set of $O(\epsilon)$-Nash equilibria is contained in $B$.

Since $\bar{X}^T=\sum_{t=1}^T X^t/T$ converges to the set of $O(\epsilon)-$Nash equilibrium  with probability 1, it also converges to $B$. 
By selection of $B$, $s_1$ and $s'_1$, this implies that $\lim_{T\to \infty}\langle e_{s_1}, A^{(12)} \bar{X}_2^T\rangle> \frac{2\delta}{3}$ while $\lim_{T\to \infty}\langle e_{s_2}, A^{(12)} \bar{X}_2^T\rangle< \frac{\delta}{3}$
By Lemma \ref{lem:equiv}, the limit of the ratio between playing $s_1$ and $s'_1$ is given by:

\begin{align*}
    \frac{X_{is_i}^T}{X_{is'_i}^T}&=\exp{\left(\eta_i \cdot T\cdot (Y_{is_i}^T-Y_{is_i'}^T)\right)}\\
    &=\exp{\left(\eta_i \cdot T\cdot (Y_{is_i}^0+\langle e_{s_1},\sum_{t=1}^TA^{(12)}X_2^t\rangle-Y_{is'_i}^0-\langle e_{s'_1},\sum_{t=1}^TA^{(12)}X_2^t\rangle) \right)}\\
    &=\exp{\left(\eta_i \cdot T\cdot (Y_{is_i}^0-Y_{is_i'}^0)\right)}\cdot \exp{\left(\eta_i \cdot T\cdot (\langle e_{s_1},\sum_{t=1}^TA^{(12)}X_2^t\rangle-\langle e_{s'_1},\sum_{t=1}^TA^{(12)}X_2^t\rangle) \right)}\\
    &=\exp{\left(\eta_i \cdot T\cdot (Y_{is_i}^0-Y_{is_i'}^0)\right)}\cdot \exp{\left(\eta_i \cdot T^2\cdot (\langle e_{s_1},A^{(12)}\bar{X}_2^T\rangle-\langle e_{s'_1},A^{(12)}\bar{X}_2^t\rangle) \right)}\\
    &\to 0 \ as  \ T\to \infty \ w.p. \ 1.
\end{align*}

Since $\langle e_{s_1},A^{(12)}\bar{X}_2^T\rangle<\frac{2\delta}{3}$ and $\langle e_{s'_1},A^{(12)}\bar{X}_2^T\rangle>\frac{\delta}{3}$. 
Finally, observe that $\frac{X_{is_i}^T}{X_{is'_i}^T}\geq {X_{is_i}^T}\geq 0$ and therefore $\frac{X_{is_i}^T}{X_{is'_i}^T}\to 0$ implies ${X_{is_i}^T}\to 0$ and agent 1's strategy will converge to the boundary (with $x_{is_i}\approx 0$) with probability 1. 
\end{proof}

\section{Feller Chain in the Primal-space Markov Chain}\label{app:Feller}

\Feller*

\begin{proof}
By Definition \ref{def:Feller}, we must show that $P(x^{(n)}, \cdot) \Rightarrow P(x,\cdot)$ as $x_n \to x$. 
It suffices to show that for all continuous and bounded functions $g: {\cal X} \to \mathbb{R}$ that the expectation of $g$ with respect to the measure $P(x_n, \cdot)$ converges to the expectation of $g$ with respect to the measure $P(x,\cdot)$ (see e.g., Theorem 2.1 in \cite{billingsley1968convergence}).

Let $g$ be such a function. 
Recall that the probability transition kernel is
\begin{align*}
	P(x^{t-1},X^t=x^t)= \sum_{{s \in {\cal S}:}\atop{R(s,x^{t-1})=x^t}}\prod_{i\in {\cal N}} x_{is_i}^{t-1}.\tag{Probability Transition Kernel for \ref{eqn:S.MWU}}
\end{align*}
Since ${\cal S}$ is finite, $P(x^{t-1},x)=0$  for all but a finite number of $x\in {\cal X}$.  
Thus, the limit of the expectation is given by
\begin{align*}
	\lim_{n\to \infty}E_{P(x^{(n)}, \cdot)}[g(z)]
	=&
	\lim_{n\to \infty}\sum_{{z\in {\cal X}:}\atop{P(x^{(n)}, z)>0}}P(x^{(n)}, z) \cdot g(z)\\
	=& \lim_{n\to \infty}\sum_{{z\in {\cal X}:}\atop{P(x^{(n)}, z)>0}}\sum_{{s \in {\cal S}:}\atop{R(s,x^{(n)})=z}}\prod_{i\in {\cal N}} x_{is_i}^{(n)} \cdot g(z)\\
	=& \lim_{n\to \infty}\sum_{s \in {\cal S}}\prod_{i\in {\cal N}} x_{is_i}^{(n)} \cdot g(R(s,x^{(n)}))\\
	=& \sum_{s \in {\cal S}}\prod_{i\in {\cal N}} x_{is_i} \cdot g(R(s,x))
\end{align*}
since the resulting expectation is continuous and bounded in $x^{(n)}$. 
Thus, 
\begin{align*}
\lim_{n\to \infty}E_{P(x^{(n)}, \cdot)}[g(z)]&= \sum_{s \in {\cal S}}\prod_{i\in {\cal N}} x_{is_i} \cdot g(R(s,x))\\&=E_{P(x, \cdot)}[g(z)]
\end{align*}
following the same steps in reverse. 
As such, $P(x^{(n)}, \cdot) \Rightarrow P(x,\cdot)$ as $x_n \to x$ and (\ref{eqn:S.MWU}) forms a Feller chain. 
\end{proof}

\section{Stationary Distributions in the Primal-space}\label{app:stationary}

\Stationary*

\begin{proof}
	
	$(\Leftarrow)$ Let $\pi$ be such that $\pi(x)>0$ only if $x\in {\cal X}$ is a pure strategy. 
	This implies that if $\pi(x)>0$, then $x$ is in the form $x=(e_{s_1}, e_{s_2},\dots, e_{s_N})$ for some $s\in {\cal S}$ and the strategy $s$ is selected from the distribution $x$ with probability $1$ and updating $x$ with (\ref{eqn:S.MWU}) yields the strategy $x$ and therefore $P(x, X^1=x)=1$. 
	This hold for all $x$ where $\pi(x)>0$ and $\pi P = \pi$ implying $\pi$ is a stationary distribution. 
	
	$(\Rightarrow)$ 
	Let $\pi$ be any stationary measure and let $\bar{\cal X}$ be a face of ${\cal X}$ with support $\bar{\cal S}\subseteq {\cal S}$, i.e., $\bar{\cal X}=\{x\in {\cal X}: x_{i\bar{s}_i}>0\  \forall \bar{s}_i\in \bar{\cal S}_i \ \forall i=1,...,N \}$. 
	With this definition, ${\cal X}$ can be expressed as a union of the relative interior of its faces,
	i.e.,
	\begin{align*}
	{\cal X}&= \bigcup_{\bar{\cal S}\subseteq {\cal S}}\left\{ rel.int(\bar{\cal X}) \right\}\\
	&= \left\{ \cup_{s\in {\cal S}} (e_{s_1},...,e_{s_n}) \right\}\bigcup
	\left\{\cup_{\bar{\cal S}\subseteq {\cal S}: \sum_{i=1}^N|\bar{\cal S}_i|\geq N+1 \ \forall i} rel.int(\bar{\cal X}) \right\}
	\end{align*}
	where the equality follows after separating the 0-dimensional faces that correspond to pure strategies. 
	We now show that $\pi(rel.int(\bar{\cal X}))=0$ if $\bar{\cal X}$ is not a pure strategy.

	By definition of (\ref{eqn:S.MWU}), $P:  rel.int(\bar{\cal X}) \to  rel.int(\bar{\cal X})$ since $X_i^{t}>0$ if and only if $X_i^{t-1}>0$.
	Thus, if $\pi(rel.int(\bar{\cal X}))>0$, then $\pi(rel.int(\bar{\cal X}))>0$ induces an invariant function (not necessarily a measure) on $rel.int(\bar{\cal X})$.
	
	Consider the game induced on the face $\bar{\cal S}$, $\bar{\Gamma}=\Gamma (\cal N, \bar{\cal S}, A)$.  
	Since the elements of $A$ are distinct, $\bar{\Gamma}$ is non-trivial whenever ${\cal X}$ is not 0-dimensional, i.e., when $\bar{\cal X}$ is not a single pure strategy. 
	We now break the problem into two cases depending on whether $\bar{\Gamma}$ has an interior Nash equilibrium.

	First, suppose $\bar{\Gamma}$ has a Nash equilibrium $\bar{x}^*\in rel.int(\bar{\cal X})$. 
	We will now use Theorem \ref{thm:Drift} to show strategies drift away from $x^*$ implying there is no stationary probability in $rel.int(\bar{\cal X})$.
	Let $\bar{h}_i=\sum_{s_i\in \bar{\cal S}_i} x_{is_i}\ln x_{is_i}$.
	This definition produces the same update rule for $x\in \bar{\cal X}$ using (\ref{eqn:S.MWU}) and the same probability transition kernel on $\bar{\cal X}$. 
	More importantly, it allows us to use the definition (\ref{eqn:S.FTRL}) and apply Theorem \ref{thm:Drift}. 
	Let $\bar{Y^t}$ and $\bar{h}^*$ be the resulting updates from (\ref{eqn:S.FTRL}) on $\bar{\Gamma}$.
	In the proof of Theorem \ref{thm:Drift}, $\bar{h}^*_i$ is strictly convex for all $y$ and therefore $E[F_{\bar{h}}(\bar{x}^*||\bar{Y}^t| \bar{Y}^{t-1}=\bar{y}^{t-1})]> F_{\bar{h}}(\bar{x}^*||\bar{y}^{t-1})$, i.e., the Bregman divergence between $\bar{X}^t$ and $\bar{x}^*$ is expected to increase. 
	Since $\pi$ is an invariant function on $rel.int(\bar{\cal X})$, the expected value of the Bregman divergence, with respect to $\pi$, does not change and therefore $\pi(rel.int(\bar{\cal X}))=0$.

	Next, suppose $\bar{\Gamma}$ does not have an interior Nash equilibrium. 
	By Theorem \ref{thm:noninterior}, strategies converge to the boundary and therefore $\lim_{t\to \infty} (\pi P^t)(B)\to 0$ for all compact $B\subseteq rel.int(\bar{\cal X})$.
	Since $\pi$ is an invariant function,  $\lim_{t\to \infty} (\pi P^t)(B)=\pi(B)$ and therefore $\pi(B)=0$. 
	Thus, $\pi (rel.int(\bar{\cal X}))=0$. 
	
	As a result, $\pi(rel.int{\bar X})=0$ whenever ${\bar X}$ is not a pure strategy completing the second direction thereby completing the proof the theorem. 
\end{proof}

\section{Agents Play Pure Strategies}\label{app:convergence}

\Convergence*

Theorem \ref{thm:Convergence} follows from Theorems \ref{thm:Feller} and \ref{thm:Stationary} and the following result.

	\begin{lemma}[Theorem 12.3.2 in \cite{douc2018markov}]\label{lem:Fel_conv}
		Let $X$ be a Feller chain taking values on the compact metric space $({\cal X}, \rho)$.
		For an arbitrary $\mu \in {\cal B}({\cal X})$, let 
		\begin{align*}\mu_T=\frac{1}{T}\sum_{t=0}^{T-1} (\mu P^t)(\cdot)
		\end{align*}
		be the time-average of the Feller process. 
		Then
		\begin{enumerate}
			\item The sequence $\{\mu_T\}_{T=1}^\infty$ has a convergent subsequence. 
			\item Every convergent subsequence of $\{\mu_T\}_{T=1}^\infty$ converges to a stationary distribution of $X$.  
		\end{enumerate}
	\end{lemma}

\begin{proof}[Proof of Theorem \ref{thm:Convergence}]
By Theorem \ref{thm:Feller}, $\{X^t\}_{t=1}^\infty=\{X^0P^t\}_{t=1}^\infty$ corresponds to a Feller chain.
By Lemma \ref{lem:Fel_conv}, $\{\bar X^t\}_{t=1}^\infty$ must have a convergent subsequence. 
Moreover, every convergent subsequence converges to a stationary distribution. 
By Theorem \ref{thm:Stationary}, a distribution is stationary if and only if it is a mixture over the set of the pure strategies thereby completing the proof of the theorem. 
\end{proof}
\end{document}